\documentclass[onecolumn,draftclsnofoot,12pt]{IEEEtran}

\usepackage[font=small]{caption}
\usepackage{comment}
\usepackage[normalem]{ulem}
\usepackage{xcolor}
\usepackage{url}
\ifCLASSINFOpdf
\usepackage[pdftex]{graphicx}
\usepackage{bm}
\DeclareGraphicsExtensions{.pdf}
\else
\fi
\graphicspath{{Figures/}}
\usepackage{mathtools}
\usepackage{graphicx}
\usepackage{amsthm}
\usepackage{cite}
\usepackage{color}
\usepackage{bbm}
\usepackage{amsmath,amssymb,amsfonts}
\usepackage{subcaption} 
\DeclarePairedDelimiter\floor{\lfloor}{\rfloor}

\newcommand{\ie}{\textit {i.e., }}

\newcommand{\m}{\mathrm{m}}
\newcommand{\f}{\mathrm{f}}
\newcommand{\x}{\pmb{x}}
\newcommand{\y}{\pmb{y}}
\newcommand{\z}{\pmb{z}}

\newcommand{\expect}{\mathbb{E}}

\newcommand{\dd}{\mathrm{d}}

\newcommand{\tb}{T_{\mathrm{b}}}
\newcommand{\hp}[1]{\widetilde{\mathsf{p}}_{#1}}
\newcommand{\ehp}[1]{\mathsf{p}_{#1}}
\newcommand{\hpdiff}[1]{\mathsf{q}_{#1}}
\newcommand{\lhpdiff}[1]{\mathcal{Q}_{#1}}
\newcommand{\hpone}{\overline{\mathsf{p}}}
\newcommand{\lhp}[1]{\widetilde{\mathcal{P}}_{#1}}
\newcommand{\lehp}[1]{\mathcal{P}_{#1}}
\newcommand{\lhpone}{\overline{\mathcal{P}}}

\newcommand{\Amat}{\widetilde{\mathcal{A}}}
\newcommand{\Aemat}{\mathcal{A}}
\newcommand{\erfc}{\mathrm{erfc}}
\newcommand{\Pmat}{\widetilde{\mathsf{P}}}
\newcommand{\bPmat}{\mathbf{P}}
\newcommand{\Pemat}{\mathsf{P}}

\newcommand{\Rval}{R^{'}}
\newcommand{\expectA}[1]{\mathbb{E}\left[#1\right]}

\renewcommand{\bPmat}{\bm{\mathcal{P}}}

\newtheorem{theorem}{Theorem}

\newtheorem{lemma}{Lemma}
\newtheorem{coro}{Corollary}[theorem]

\newcommand{\figsize}{3.4in}

\newcommand{\invlaplace}[2]{\mathcal{L}^{-1}_{#1}\left[#2\right]}

\usepackage{graphics}
\begin{document}
	\title{Channel Characterization and Performance of a 3-D Molecular Communication System with Multiple Fully-Absorbing Receivers}
		\author{Nithin V. Sabu, Abhishek K. Gupta, Neeraj Varshney and Anshuman Jindal
		\thanks{ N. V. Sabu, A. K. Gupta and A. Jindal are with Indian Institute of Technology Kanpur, Kanpur UP 208016, India (Email:{ \{nithinvs,gkrabhi,anshuji\}@iitk.ac.in}). N. Varshney is with the Wireless Networks Division, National Institute of Standards and Technology, Gaithersburg, MD 20899 USA (Email: {neerajv@ieee.org}). 
			This research was supported by the Science and Engineering
			Research Board (India) under the grant SRG/2019/001459 and IITK under the grant IITK/2017/157. A part of this paper will be presented at the IEEE ICC Workshop,  Seoul, South Korea, May 2022 \cite{sabu2022c}.}}
	\maketitle
	\begin{abstract}
Molecular communication (MC) can enable the transfer of information between nanomachines using molecules as the information carrier. In MC systems, multiple receiver nanomachines often co-exist in the same communication channel to serve common or different purposes. However, the analytical channel model for a system with multiple fully absorbing receivers (FARs), which is significantly different from the single FAR system due to the mutual influence of FARs, does not exist in the literature. The analytical channel model is essential in analyzing systems with multiple FARs, including MIMO, SIMO, and cognitive molecular communication systems. In this work, we derive an analytical expression for the hitting probability of a molecule emitted from a point source on each FAR in a diffusion-based MC system with $ N $ FARs. Using these expressions, we derive the channel model for a SIMO system with a single transmitter and multiple FARs arranged in a uniform circular array (UCA). We then analyze the communication performance of this SIMO system under different cooperative detection schemes and develop several interesting insights.
	\end{abstract}


\section{Introduction}
 Molecular communication (MC) is a promising solution for enabling the transfer of information in several scenarios considered impractical for conventional electromagnetic wave-based communication (EMC). For example, MC is expected to perform better than EMC inside tunnels, in saline environments, and inside the human body \cite{farsad23,guo2015}. In MC, the transmitter conveys the information to the receiver by utilizing molecules, termed as \textit{information molecules} (IMs), as the carrier of the information. Among different propagation mechanisms, the diffusion-based propagation mechanism is the most studied one in the literature due to the ease of mathematical modeling, energy efficiency, and lack of requirement for communication infrastructure for carrying IMs from the transmitter to the receiver.

In MC, the transmitters and receivers are usually nanomachines with nanoscale functional parts. The functionality of a single nanomachine is limited to simple actuation, sensing, and storage \cite{nakano2013}. However, complex tasks like target detection and event sensing can be performed via the cooperation of multiple nanomachines. Such systems with multiple nanomachines working together can help to realize the internet of nano \cite{akyildiz2010}, and bio-nano things (IoNT and IoBNT) \cite{akyildiz2015}.
 
 \subsection{Related Work}\label{relw}
For an MC system with a point transmitter and a fully-absorbing receiver (FAR), the analytical expression for the hitting probability was derived in \cite{yilmaz2014,schulten2000lectures}. As mentioned above, multiple receiver nanomachines can exist in an MC to enable complex tasks. Therefore, the channel model for an MC system when multiple receivers are present is of crucial importance. Many works \cite{koo2016, damrath2018,lee2017,ahuja2021,gursoy2019a} employed fitting models, machine learning-based approaches, and particle-based simulations to obtain the channel model. However, an exact analytical equation for the hitting probability for MC systems involving multiple fully absorbing receivers in 3-D is not present in the past literature. \cite{twoway} derived an approximate analytical expression for hitting probability (with two unknowns to be computed numerically) for an MCvD system with two FARs. In our past work \cite{sabu2020a}, we have derived an approximate analytical expression for the hitting probability of an IM on each of the FAR of the same size in a $ 2-$FAR system. The approximate expression was simple enough so that several design insights were discussed in \cite{sabu2020a}. Later, this work was extended to derive the hitting probability of IM on each FAR with different sizes \cite{sabu2021}. For MC systems involving more than two FARs, an approximate expression for the hitting probability of an IM on \textit{any one} of the FARs was derived in \cite{sabu2020}. However, the expression for the hitting probability on each of the FARs is still missing from the current literature. Systems with multiple receivers may improve the performance of the MC systems \cite{koo2016, damrath2018,lee2017,ahuja2021,gursoy2019a}. In MC systems with multiple receivers, cooperative detection can be used to improve the error performance  \cite{yutingfang2017}. In a cognitive MC system \cite{sabu2021}, multiple primary and secondary receivers may be present.  The analysis of these systems with multiple FARs will require a channel model to start. Therefore, it is essential to analyze and characterize the channel in MC systems with multiple receivers. According to the authors knowledge, no work in the literature provides an analytical expression for the 3-D channel model of systems with multiple FARs (\ie more than two FARs) using the same type of IM. Obtaining hitting probability using machine learning and particle-based simulation methods \cite{koo2016, damrath2018,lee2017,ahuja2021,gursoy2019a}  is time-consuming and inconvenient. In this work, we bridge this gap by deriving the analytical hitting probability on each FAR in a multi-FAR system and characterizing the channel in this system. 
 
 \subsection{Contributions}
In this work, we consider a 3-D diffusion-based MC system with multiple FARs and derive the hitting probability of an IM on each of the FARs. We also validate the accuracy of the derived equations using particle-based simulations. In contrast to \cite{sabu2020a} which only considers two FARs, this work considers $N$ FARs and includes the necessary modifications in the differential equations to incorporate the mutual influence of all the FARs. We also analyze a single-input multiple-output (SIMO) system comprising a uniform circular array (UCA) of FARs based on cooperative detection schemes. Several important design insights are also provided based on the derived expressions.

The novel contributions of this work are listed below.

	\begin{itemize}
\item First, we present a set of analytical equations for the exact hitting probability at each FAR in a system with $ N $ FARs in 3-D for any arbitrary value of $N$, which involves the distribution of hitting point on each FAR's surface. We then present a more tractable expression for the hitting probability by approximating these hitting points. We then present a recursive method that can be convenient when the hitting probability for subsystems is already known. We also present error bound and several simplified expressions under special cases. 
\item  To demonstrate the applicability of the derived expressions, we consider a SIMO communication system comprised of a point transmitter and a UCA of FARs as a compound receiver and present the  analytical expressions for the hitting probability. We also derive the upper and lower bounds for the exact channel response and present various design insights.
\item Using these expressions, we derive the channel model for the considered SIMO system and present various design insights, including the asymptotic signal gain of UCA. 
\item Using the derived channel model, we analyze the communication performance of the considered SIMO system with a cooperative detection scheme and compare the performance of both soft decision and hard decision rules. We also quantify how the bit error performance changes when the number of FARs increases.
\end{itemize}



\section{System Model}\label{secsys}

    In this work, we consider a diffusion based MC system with a point transmitter and $N$ spherical FARs located  at different positions in a 3D medium as shown in Fig. \ref{Fig:sm} (a). Let the transmitter be located at the origin and $N$ FARs of the same radius $ a $ be at positions $ \x_{1}, \x_{2} $, $ \x_{3} $,..., $ \x_{N} $ respectively. The distance between the transmitter and the $i$th FAR FAR$_i$ located at $ \x_i $ is represented by $ r_i $. The angular difference between the centers of the 
FAR$_i$ and FAR$_j$
is denoted by $ \phi_{ij} $. Let $ \y_i $ denote the closest point of FAR$ _i $ from the transmitter.  Also, $ R_{ij} $ represents the distance between the closest point of FAR$_i  $ from the origin (\ie $ \y_i $) and the center of FAR$ _j$. \ie $ R_{ij}=\sqrt{(r_i-a)^2+r_j^2-2(r_i-a)r_j\cos\phi_{ij}} $.
The transmitter emits IMs, which propagate through the medium via Brownian motion \cite{einstein2011}.  The diffusion coefficient of the IM is represented by $ D $, which is assumed to be constant over space and time. The IMs are detected by a FAR when they hit the surface of that particular receiver.

	\begin{figure}[ht]
		\centering   
		\includegraphics[width=0.4\linewidth]{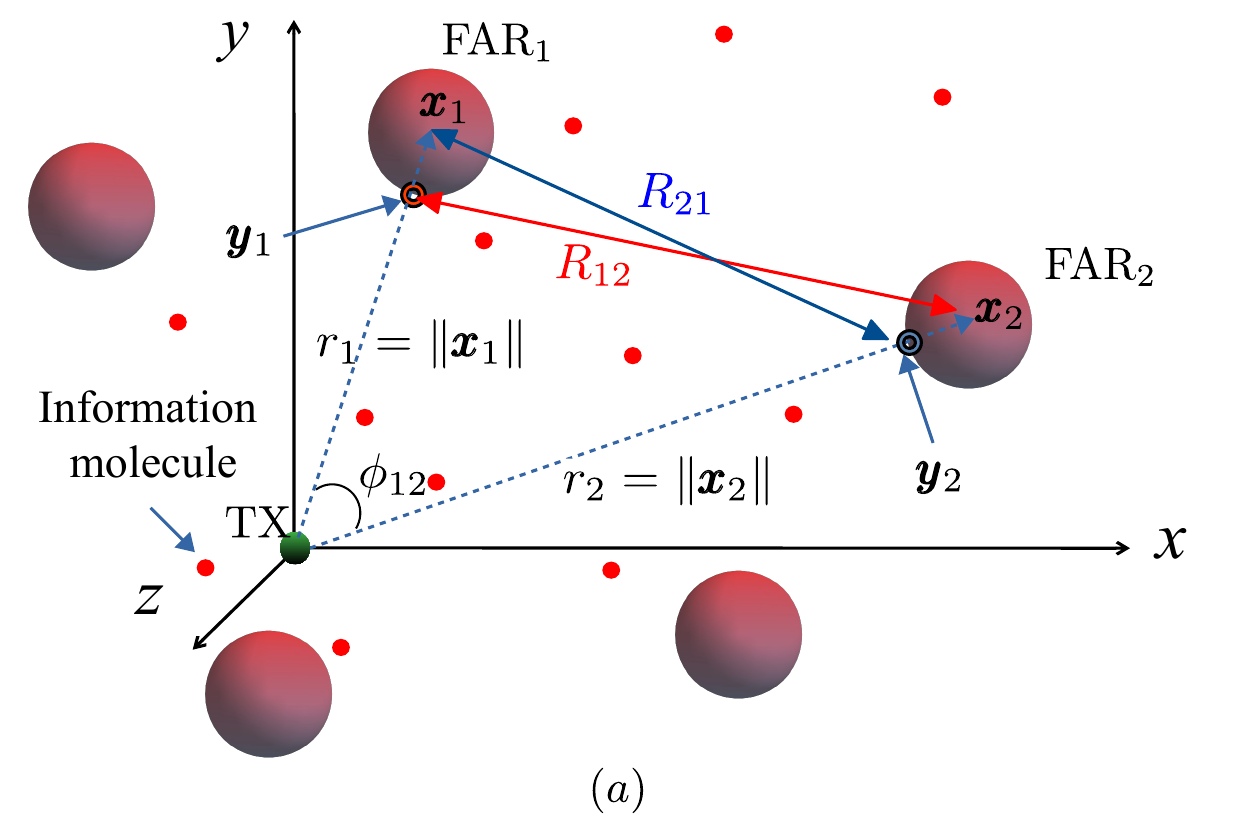}\includegraphics[width=0.4\linewidth]{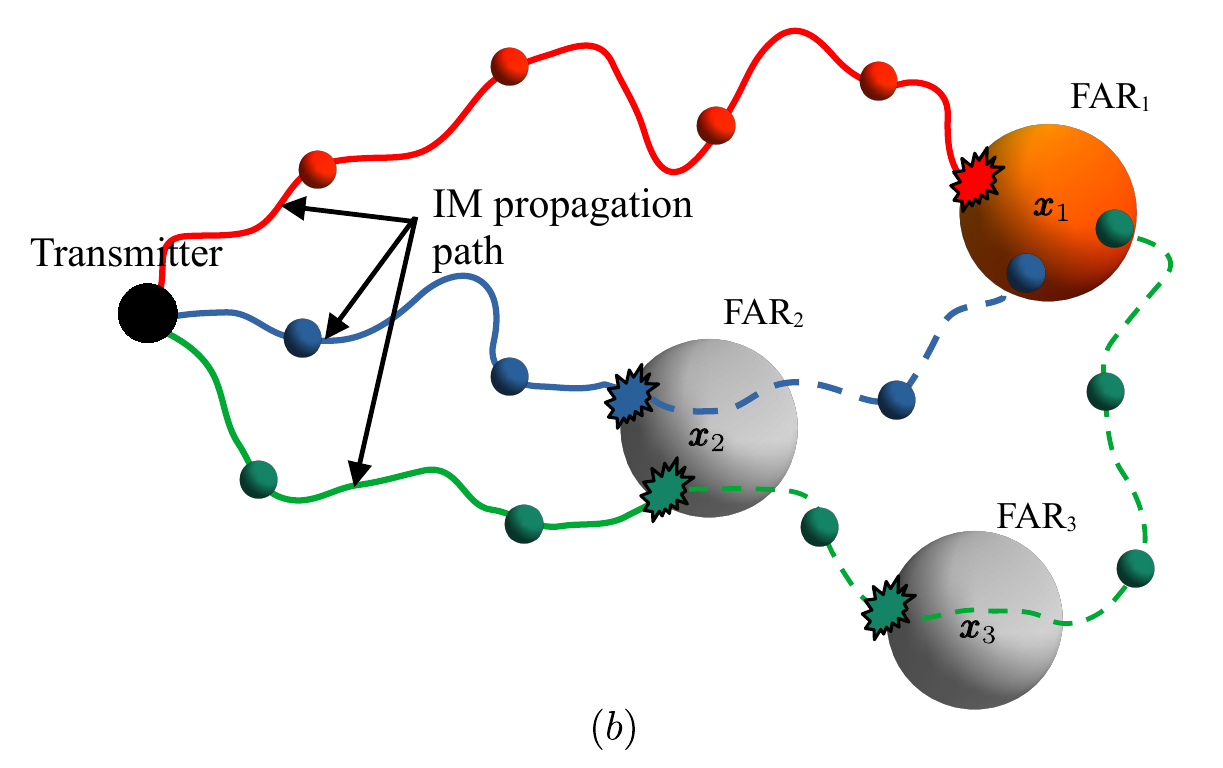}
		\caption{ An illustration showing (a) a 3D MC  system with a point transmitter at the origin and FARs located at positions $\x_{1},\ \x_{2},\cdots, \x_{N}$ in $ \mathbb{R}^3 $,  
		(b) different paths a molecule can take to reach a receiver FAR$_1$. The path of IMs hitting FAR$_1$ first is represented by the red path. The paths that hit any other FAR first and then hit FAR$_1$ if the other FARs are not present are represented by the blue and green paths.
		}
		\label{Fig:sm}
	\end{figure}

	When there are no other FARs present except the $i$th FAR, the hitting probability $ \hpone(t,r_i) $ of an IM at the surface of the $ i $th FAR within time $t$
is given by \cite{yilmaz2014,schulten2000lectures}
	\begin{align}
		\hpone(t,r_i)=\frac{a}{r_i}\erfc\left(\frac{r_i-a}{\sqrt{4Dt}}\right).\label{eq1rx}
	\end{align}
	 
Fig. \ref{Fig:sm} (b) shows different paths a molecule can take to reach a receiver FAR$_1$. In the absence of other FARs, all red/green/blue paths shown lead to hitting FAR$_1$. When other FARs are present, it causes competition in capturing the IMs among the FARs. For example, in Fig. \ref{Fig:sm} (b), IMs following only the red path will hit FAR$_1$. Therefore, the hitting probability reduces in the presence of other FARs. For the case when there are only two FARs located at $ \x_1 $ and $ \x_2 $  present in the medium  (\ie $ N=2 $), an approximate analytical hitting probability until time $ t $ of an IM on FAR$ _1 $ was given in our previous work \cite{sabu2020a} as
	\begin{align}
		\hp{1}(t,\x_1\mid \x_2)=
		&\sum_{n=0}^{\infty} \frac{a^{2 n}}{R_{12}^{n} R_{21}^{n}}\left[\frac{a}{r_{1}} \erfc\left(\frac{r_{1}{-}a+n\left(R_{21}{-}a\right)+n\left(R_{12}{-}a\right)}{\sqrt{4 D t}}\right)\right. \nonumber\\
		&\left.-\frac{a^{2}}{r_{2} R_{21}} \erfc\left(\frac{r_{2}{-}a+(n+1)\left(R_{21}{-}a\right)+n\left(R_{12}{-}a\right)}{\sqrt{4 D t}}\right)\right]\le \hpone(t,r_i).
		\label{eq2rx}
	\end{align}
	
    The hitting probability reduces further with an increase in $N$ due to the increase in the mutual influence of FARs, which is caused by the increasing competition in capturing the IMs by the FARs. However, the expression for the general $N$- FAR case is not available in the literature. The hitting probability expressions for $ N>2 $ can be used to analyze SIMO, multiple-input multiple-output (MIMO), and cognitive molecular communication systems. In the following sections, we derive the hitting probability (or fraction of IMs absorbed) at each of the FARs for an $ N-$FAR system. To demonstrate the applicability, the derived equations are then applied to a SIMO system based on co-operative detection to derive its BER performance under various receiver combining schemes, and the error performance is analyzed.

Let us denote the exact hitting probability until time $ t $ of an IM on FAR$ _i $ in the presence of the other $ N-1 $ FARs by $ \ehp{}\left(t,\x_i\mid \{\x_j\}_{j=1,\ j\neq i}^N\right) $ or just $ \ehp{i}(t) $ when there is no ambiguity. The approximate hitting probability until time $ t $ of an IM on FAR$ _i $ in the presence of the other $ N-1 $ FARs is denoted by $ \hp{}\left(t,\x_i\mid \{\x_j\}_{j=1,\ j\neq i}^N\right) $ or just $ \hp{i}(t) $ when there is no ambiguity. Without the loss of generality, we will take $ i=1 $ and derive the hitting probability at FAR$ _1 $ denoted by $ \hp{}\left(t,\x_1\mid \{\x_j\}_{j=2}^N\right) $ or just $ \hp{1}(t) $. Let $ \mathcal{L}\left[.\right] $ and $ \mathcal{L}^{-1}\left[.\right] $ denote the Laplace transform and inverse Laplace transform, respectively.
	
\section{Hitting Probability for an $N-$FAR System}\label{nfarsec}
In this section, we derive the  hitting probability expression for an  $ N-$FAR system for an arbitrary $N$. 
Note that the difference $\hpone(t,r_i) - \ehp{i}(t) $ gives the probability that an IM takes a path that would have hit FAR$_i$ in the absence of other FARs, but hits a FAR other than FAR$_i$ first when in the presence of other FARs. This event corresponds to IMs taking a green or blue path shown in Fig. \ref{Fig:sm} (b).
Let us consider a path where the IM hits the FAR$_j$ first. The hitting point is a random point $\z_j$ on the surface of FAR$_j$. This path can be divided into two segments. The first corresponds to the segment from the transmitter to the FAR$_j$ (denoted by solid line). The probability that an IM emitted from the point source at the origin first hits the surface of FAR$_{j} $ in the interval $[\tau,\tau+\dd\tau] $ is $\frac{\partial \ehp{j}(\tau)}{\partial \tau} \dd\tau$. Further, the second segment corresponds to the segment the IM would have followed to reach FAR$_i$ if it were not absorbed by FAR$_j$ in the first segment despite hitting it or by any FAR it hits during this segment. It is evident that the presence of FARs is ignored for this segment. This segment starting from $\z_j$ and ending at FAR$_i$ is denoted by dashed path. The corresponding probability that the IM reaches FAR$_i$ in $t-\tau$ time taking this path is $\expect_{z_j}\left[\hpone(t-\tau,\Rval_{ji})\right]$, where $\Rval_{ji}=\|\z_j-\x_i\|$.

 Therefore, the probability that an IM would have hit FAR$_i$ in the absence of other FARs, but hits a FAR other than FAR$_i$ first when in the presence of other FARs, is given as
\begin{align}
	\hpone(t,r_i) - \ehp{i}(t)   =&  \sum_{j=1,\ j\neq i}^{N} 
	\int_{0}^{t} \frac{\partial \ehp{j}(\tau)}{\partial \tau}\expect_{z_j}\left[\hpone(t-\tau,\Rval_{ji})\right] 	\dd\tau\label{eqtnxra} .
\end{align}
Now, taking the Laplace transform of both sides of  \eqref{eqtnxra} gives
\begin{align}
	\lhpone(s,r_i) -\lehp{i}(s)  =  \sum_{j=1,\ j\neq i}^{N} 
	s\lehp{j}(s)\expect_{z_j}\left[ \lhpone(s,\Rval_{ji})\right] , \ \ \forall\ i.\label{exsnxr} 
\end{align}
Here $\lehp{i}(s)$ and $\lhpone(s,x),\ x\in\{r_i,\Rval_{ij}\}\ $ are the Laplace transforms of $\ehp{i}(t)$ and $\hpone(t,x)$ respectively. Note that,
\begin{align}
	\lhpone(s,x)=\frac{a}{sx}\exp\left(-\left(x-a\right)\sqrt{\frac{s}{D}}\right).\label{lapeq}
\end{align}  
Now stacking \eqref{exsnxr} for all $ i $ and solving further for	$ \Pemat(t)=\left[ \begin{array}{c} \ehp{1}(t),\  \ehp{2}(t), \cdots ,\ \ehp{N}(t) \end{array} \right]^{\mathrm{T}} $ gives the theorem given below.
\begin{theorem}\label{exfar}
	In an MC system with $N$ FARs, the exact probability that an IM emitted from the point source at the origin  hits each of the $\mathrm{FAR}$s within time $ t $ in the presence of other $N-1$ FARs is 
	\begin{align}
		\Pemat(t)=\mathcal{L}^{-1}\left(\Aemat^{-1}(s)\bPmat(s)\right),\label{eqexact}
	\end{align}
	where
	\begin{align*}
			\Aemat(s)&=
	{\small	\begin{bmatrix} 
			1 & s\expect_{z_2}\left[\lhpone(s,\Rval_{21})\right]  & \cdots & s\expect_{z_N}\left[\lhpone(s,\Rval_{N1})\right] 	\\
			s\expect_{z_1}\left[\lhpone(s,\Rval_{12})\right]& 1 & \cdots & s\expect_{z_N}\left[\lhpone(s,\Rval_{N2})\right]\\ 
			\vdots &\vdots &\ddots&\vdots \\
			s\expect_{z_1}\left[\lhpone(s,\Rval_{1N})\right] & s\expect_{z_2}\left[\lhpone(s,\Rval_{2N})\right] & \dots & 1 \\
		\end{bmatrix}}\nonumber\\
		 \text{ and } \bPmat(s)&=\left[ \begin{array}{cc} \lhpone(s,r_1),\ 	\lhpone(s,r_2)  ,\ \cdots ,\ \lhpone(s,r_{N}) \end{array} \right]^{\mathrm{T}}.
	\end{align*}
\end{theorem}
\begin{proof}
	The system of equations given in \eqref{exsnxr} can be represented in matrix form as  $\Aemat(s)\Pemat(s)=\bPmat(s)$. Hence, the solution can be obtained as
	\begin{align}
		\Pemat(s)=\Aemat^{-1}(s)\bPmat_{}(s).\label{exmateqn}
	\end{align}
	Finally, taking the inverse Laplace transform of both sides of \eqref{exmateqn} gives \eqref{eqexact}. 
\end{proof}

 	To solve the expression in \eqref{eqexact} further to obtain the exact hitting probability, we require the distribution of hitting point $ \z_j $ at the surface of FAR$_j$. However, the distribution is difficult to obtain, and hence, we need to approximate the hitting point.
 	 
 	 The actual hitting point $ \z_j $ of an IM at a FAR$ _j $ can be approximated by the closest point on the surface of the FARs as seen from the transmitter at the origin (\ie $ \y_j $). Let us denote $ \Pmat(t)=\left[  \hp{1}(t), \ \hp{2}(t),\  \cdots, \ \hp{N}(t) \right]^{\mathrm{T}} $. The hitting probability equation using the approximation mentioned above is given in the theorem below.
 	 \begin{theorem}\label{tnfar}
 	 	
 	 	In an MC system with $N$ FARs, the probability that an IM emitted from the point source at the origin  hits each of the $\mathrm{FAR}$s within time $ t $ in the presence of other $N-1$ FARs is approximately given as 
 	 	\begin{align}
 	 		\Pemat(t)\approx\Pmat(t)=\mathcal{L}^{-1}\left(\Amat^{-1}(s)\bPmat(s)\right), \text{where }	\Amat(s)=
 	 		{\small\begin{bmatrix} 
 	 			1 & s\lhpone(s,R_{21}) &  \cdots & s\lhpone(s,R_{N1}) 	\\
 	 			s\lhpone(s,R_{12}) & 1 &  \cdots & s\lhpone(s,R_{N2})\\ 
 	 			\vdots &\vdots &\ddots&\vdots \\
 	 			s\lhpone(s,R_{1N}) & s\lhpone(s,R_{2N}) &  \dots & 1 \\
 	 		\end{bmatrix}}.\label{eqnxr}
 	 	\end{align}
 	 \end{theorem}
 	 \begin{proof}
 	 	See Appendix \ref{anxr}.
 	 \end{proof}

\subsection{ Special Case: Systems with $N=3$ FARs}
We can simplify the derived expression \eqref{eqnxr} for special cases such as for specific values of $ N $. In this section, we consider the special case with $ N=3 $ FARs. The following corollary provides an approximate analytical equation for $\hp{1} (t)$.
\begin{coro}\label{t3rx}
	In an MC system with three FARs located at $\x_1,\ \x_2$ and $\x_3$, the probability that an IM emitted from the point source at the origin hits  $\mathrm{FAR} _1 $ within time $ t $ in the presence of other two FARs is approximately given by
	\begin{align}
		\ehp{1}(t)\approx\hp{1}(t)&=\mathcal{L}^{-1}\left[\lhp{1}(s)\right],\label{eq3rx}
		\text{ where } 
		\lhp{1}(s)= \frac{\lhpone(s,r_i) -s\alpha_1(s)+s^2\beta_1(s)}
		{1-s^2\gamma(s)+s^3\delta(s)}, \\
			\text{with }\alpha_1(s)&=\lhpone(s,r_2)\lhpone(s,R_{21}){+}\lhpone(s,r_3)\lhpone(s,R_{31})\nonumber\\
			\beta_1(s)&=-\lhpone(s,r_1)\lhpone(s,R_{23})\lhpone(s,R_{32})+\lhpone(s,r_2)\lhpone(s,R_{23})\nonumber\\
			&\qquad\qquad\lhpone(s,R_{31})+\lhpone(s,r_3)\lhpone(s,R_{32})\lhpone(s,R_{21})\nonumber\\
			\gamma(s)&=\lhpone(s,R_{12})\lhpone(s,R_{21}) + \lhpone(s,R_{32})\lhpone(s,R_{23})
			+\lhpone(s,R_{13}) \lhpone(s,R_{31})\nonumber\\
			\text{and }\  
			\delta(s)&=\lhpone(s,R_{12}) \lhpone(s,R_{23})\lhpone(s,R_{31})
			+ \lhpone(s,R_{13})\lhpone(s,R_{32})\lhpone(s,R_{21}).\nonumber
	\end{align}
\end{coro}
Note that, in Corollary \ref{t3rx}, $ \alpha_1(s) $ corresponds to the motion of IMs via FAR$ _2 $ to FAR$ _1 $ and via FAR$ _3 $ to FAR$ _1 $. This term double counts the IMs that went to both FAR$ _2 $ and FAR$ _3 $ (including IMs that first went to FAR$ _2 $ and IMs that first went to FAR$ _3 $). This double counting is negated by  $ \beta_1(s) $. The first term in $ \beta_1(s) $ corresponds to the IMs that went to FAR$ _1 $ first and then went to FAR$ _2 $ and FAR$ _3 $. The IMs can go in loops between FAR$ _2 $ and FAR$ _3 $ before hitting FAR$ _1 $ and such IMs are factored in by the term $ \gamma(s) $, whereas the double counting in $ \gamma(s) $ is negated using the term $ \delta(s) $.

\subsection{Mutual Influence}\label{mutinf}
The mutual influence of FARs in capturing the IM (denoted by $ \hpdiff{i}(t) $) can be  characterized by the decrease in the hitting probability of an individual FAR due to the presence of additional FARs.
	This is equal to the fraction of IMs that are supposed to hit FAR$ _i $ within time $ t $, but are hitting the other FARs first due to their presence in the same communication medium. \ie
	\begin{align}
		\hpdiff{i}(t)=\hpone(t,r_i)-\ehp{i}(t).\label{ediff}
	\end{align}
A higher value of $ \hpdiff{i}(t) $ represents that a higher number of IMs that were supposed to hit FAR$ _i $ were captured by other FARs. The following result provides an upper bound on the mutual influence.
\begin{coro}
	The upper bound for the mutual influence $ \hpdiff{i}(t) $ of FAR$ _i $ is given by 
	\begin{align}
		\hpdiff{i}(t)\leq \sum_{j=1,\ j\neq i}^{N}\frac{a^2}{r_j\left(\|\x_j-\x_i\|-a\right)}\erfc\left(\frac{r_j+\|\x_j-\x_i\|-3a}{\sqrt{4Dt}}\right)\label{mr29}.
	\end{align}
\end{coro}
\begin{proof}
	Using \eqref{ediff} and \eqref{exsnxr}, the Laplace transform of $ \hpdiff{i}(t) $ (represented by $  \lhpdiff{i}(s)  $) can be upper bounded as
	\begin{align}
		\lhpdiff{i}(s){=}\lhpone(s,r_i) -\lehp{i}(s) {= } \sum_{j=1,\ j\neq i}^{N} 
		s\lehp{j}(s) \expect_{z_j}\left[ \lhpone(s,\Rval_{ji})\right]{\stackrel{(a)}{\leq}}\sum_{j=1,\ j\neq i}^{N} 
		s\lhpone(s,r_j) \lhpone(s,\|\x_j-\x_i\|-a),\label{mr28}
	\end{align}
	where $ (a) $ is due to the inequalities $\lehp{j}(s)\leq \lhpone(s,r_j) $ and $ \expect_{z_j}\left[ \lhpone(s,\Rval_{ji})\right]\leq \lhpone(s,\|\x_j-\x_i\|-a) $. Now, taking the inverse Laplace transform of \eqref{mr28} gives \eqref{mr29}.
\end{proof}
Since the erfc(.) function decreases faster than exponential, the mutual influence reduces at least exponentially with the distance between FARs (\ie $ \|\x_j-\x_i\| $).
\begin{coro}\label{rinf}
	In an MC system with $N$ FARs located sufficiently far from each other (\ie $ R_{ij}\rightarrow \infty,\forall\  i,j $), the mutual influence of FARs perishes, and the hitting probability for an IM (fraction of IMs absorbed) on each FAR is the same as \eqref{eq1rx}.
\end{coro}
\begin{proof}
	When $R_{ij}\rightarrow \infty $, $ s\lhpone(s,R_{ij})\rightarrow 0 $. Therefore, $ \Amat(s)=\mathrm{I} $.\\
	Now,  from \eqref{eqnxr}, $ \Pmat(t)= \mathcal{L}^{-1}(\bPmat(s))$.
\end{proof}	
\subsection{Accuracy of the Hitting Probability Equations}\label{accsec}

\begin{figure}
	\centering
	\includegraphics[width=0.45\linewidth]{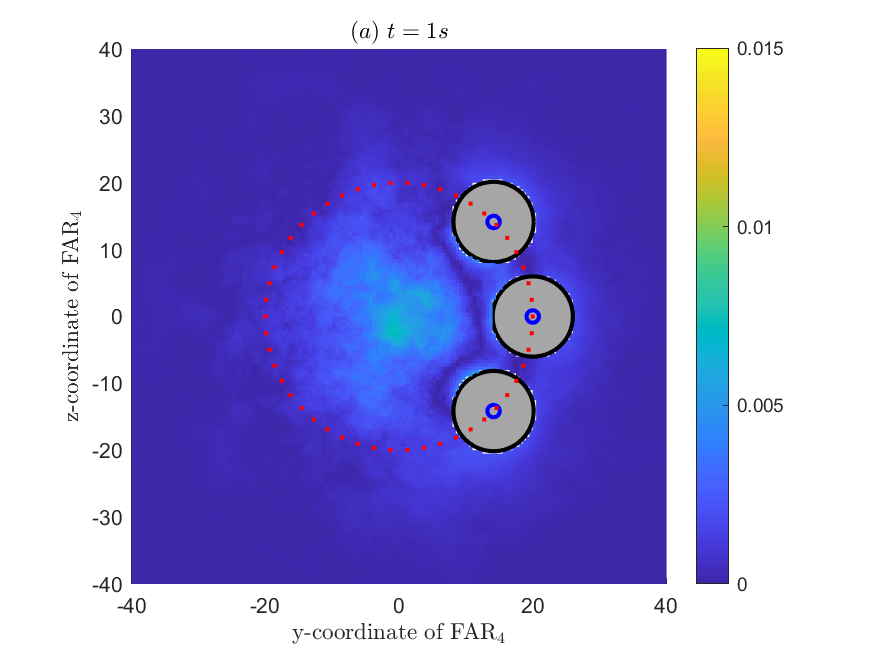}
	\includegraphics[width=0.45\linewidth]{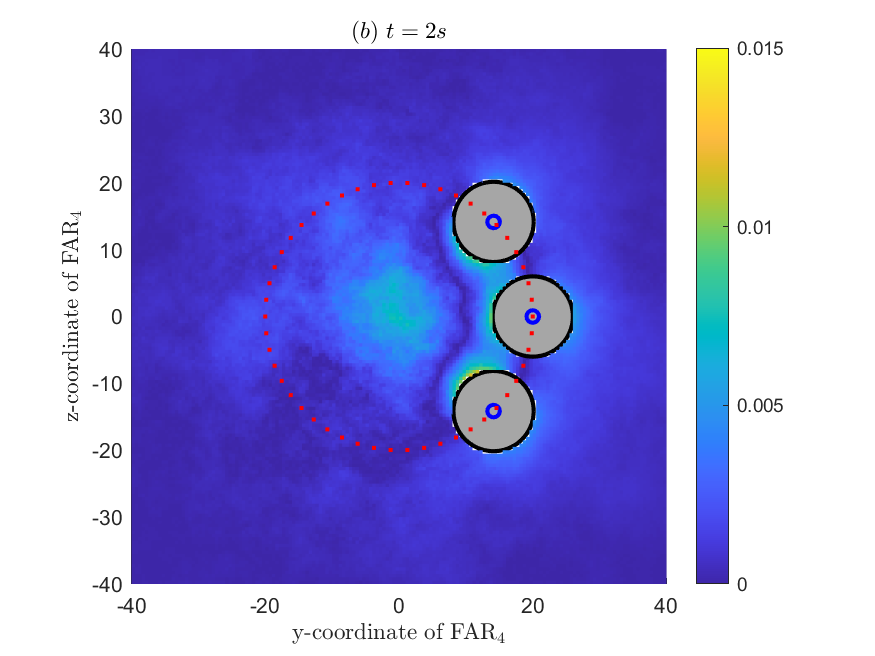}
	\caption{Variation of the absolute error of the hitting probability of FAR$ _4 $ in a $ 4- $FAR system. Parameters: $ a=3\mu m,\ D=100\mu m^2/s,\ t=1s,\ \Delta t=10^{-4}s,\ \x_1=[10,\ 20,\ 0],\ \x_2=[10,\	14.14,\	14.14],\  \x_3=[10,\	14.14,\	-14.14],\  \x_4=[10,\ y,\	z], \text{where } y,z$ is varied in $ y,z $ axis.}
	\label{fig:aberr}
\end{figure}
The derived analytical equations are validated using particle-based simulation with the help of MATLAB. The numerical inverse Laplace transform function from \cite{mcclure2021} is used to evaluate the analytical results. In the particle-based simulation, an IM is generated in each iteration, and the Brownian motion of the IM is tracked. The location of the IM in $ \mathbb{R}^3 $ is tracked and checked to see whether it hits any of the FAR. If it hits any FAR, it is removed from the environment. After many iterations, the average hitting events at each FAR are taken to get the hitting probability.  The step size $ \Delta t $ taken for the simulation is $ 10^{-4}\ s $ .

Fig. \ref{fig:aberr} shows the variation of the absolute error ($= \mid $analytical value $ - $ simulated value$ \mid $) of hitting probability at FAR$ _4 $ for a $4-$ FAR system. Here, the transmitter is at the origin, and the centers of the other FARs (blue circles) are located in a circle (dotted circle). The gray circles represent the area where FAR$_4$ cannot be located (as the FARs will overlap). It can be verified that the absolute error is minimal in most regions, even in those where FARs are close to each other. 

To investigate the approximation error in \eqref{eqnxr}, we have performed extensive simulations for various system configurations. These results are summarized in \cite{web}.

Similar to the 2-FAR system \cite{sabu2020a}, the accuracy of the derived equation can degrade if the FARs are too close to each other and the target FAR hinders the other FARs from the line of sight of the transmitter. In such scenarios, the approximating hitting point to a single point $ \y_i $ is not a good approximation because the IM may hit the FARs at a different random point.

As explained, the only source of approximation in \eqref{eqnxr} is due to the approximation of $ \z_j $ by $ \y_j $. Since $ \z_j $ and $ \y_j $ are close for small receivers ($ a\rightarrow 0 $), we can expect the approximation error to be minimal. The following lemma presents the bounds on the absolute approximate error to provide a further theoretical guarantee.
\begin{lemma}\label{eboundc}
	The approximation error in the approximate hitting probability of $ i $th FAR {\color{black}in the Laplace domain}, as given by \eqref{eqnxr}, is bounded as
	\begin{align}
		\left| E_i(s)\right| 
		\leq  \sum_{k=1}^N\sum_{j=1,\ j\neq k}^{N}	\left| c_{ik}(s)\right|  G_{kj}(s) \ s\lhpone(s,r_j),\label{abslem}
	\end{align}
	where $ c_{ik}(s) $ is the $ i,k $th element of $ \Amat^{-1}(s) $ and 
	\begin{align}
		G_{kj}(s)
		&= 
		\lhpone(s,R_{jk})
		\max\left\{\left[\frac{R_{jk}}{\|\x_j-\x_k\|-a}e^{2a\sqrt{\frac{s}{D}}}-1\right],	
		\left[{e^{2a\sqrt{\frac{s}{D}}}}{}-\frac{R_{jk}}{\|\x_j-\x_k\|+a}\right]\right\}\label{eq:errGijm}.
	\end{align}
	\begin{proof}
		See Appendix \ref{ebound}.
	\end{proof}
\end{lemma}
From \eqref{eq:errGijm}, we can verify that, as $ a\rightarrow 0 $, $ G_{kj}(s) \rightarrow 0$ and $ 	|{E_i}(s) | \rightarrow {0} $ for all $i$. Note that $ G_{kj}(s) $ goes to zero relative to $\lhpone(s,R_{jk})$. 
Therefore, this results show that the error is small relative to hitting probability when the radius of the FARs is small.
\begin{coro} \label{asympcoro}
	The asymptotic approximation error in the approximate hitting probability of $ i $th FAR is bounded as
	\begin{align}
		\left| E_i(\infty)\right|\leq \sum_{k=1}^N\sum_{j=1,\ j\neq k}^{N}	\left| c_{ik}(0)\right|\frac{a^2}{r_jR_{jk}} \max\left\{\left[\frac{R_{jk}}{\|\x_j-\x_k\|-a}-1\right],	
		\left[1-\frac{R_{jk}}{\|\x_j-\x_k\|+a}\right]\right\}. \label{asymp}
	\end{align}
	\begin{proof}
		Applying the final value theorem for the Laplace transform (\ie $ \lim_{t\rightarrow\infty} f(t)=\lim_{s\rightarrow0}sF(s)$) \cite{spiegel1992} in \eqref{abslem} gives \eqref{asymp}.
	\end{proof}
\end{coro}

\begin{figure}[h]
	\centering
	\includegraphics[width=0.6\linewidth]{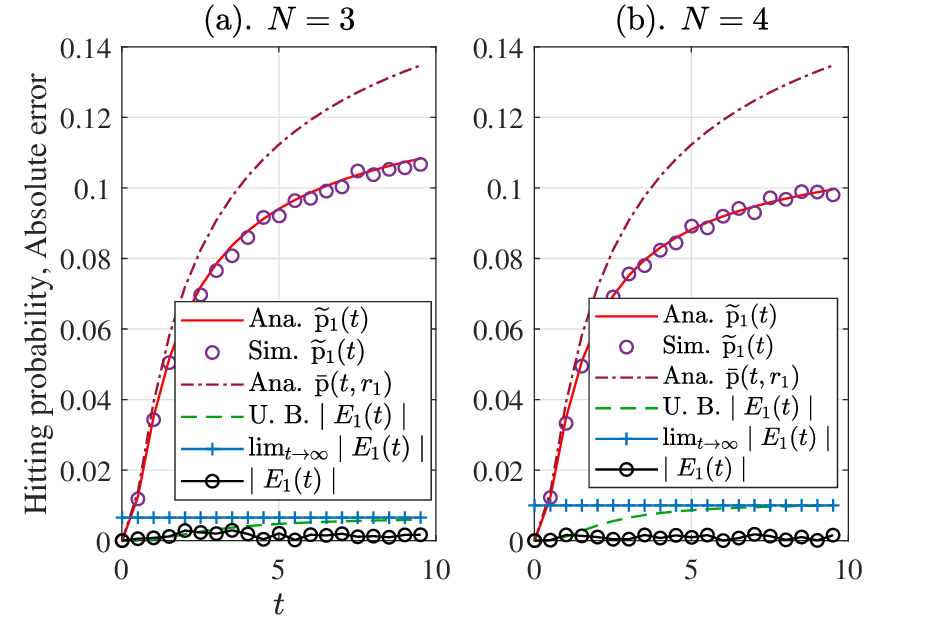}
	\caption{{\color{black}Variation of hitting probability and absolute approximation error with time for an MC system with $ N $ FARs. Parameters : $\x_i= [0,\ 20\cos(2\pi i/N),\ 20\sin(2\pi i/N)],\ N\in\{3,4\},\ D=100\mu m^2/s $.}}
	\label{fig:abserr2m1}
\end{figure}
{\color{black} Fig. \ref{fig:abserr2m1} shows the variation of hitting probability and absolute error for $ N=3 $ and $ N=4 $ FARs for an MC system with UCA of FARs. We can verify that the absolute error is minimal and the upper bound of the absolute error derived in Lemma \ref{eboundc} is valid. The accuracy of derived expressions is theoretically guaranteed by the fact that the upper bound is small in comparison to the actual hitting probability. The dash-dotted lines in the Fig. \ref{fig:abserr2m1} correspond to the hitting probability of IMs in the absence of other FARs (\ie \eqref{eq1rx}). The growing deviation between the dash-dotted lines  and solid lines with time $ t $ accounts for the increasing mutual influence of FARs with time $ t $.}

\subsection{Recursive Approach}\label{subsecrec}
We now present one more approach to derive the hitting probability. 
It is based on recursive method that gives the same expression for the hitting probability; however, it can be convenient when expressions for the sub-systems are known. 
In this method, the Laplace transform of the hitting probability for an $ N-$FAR system is expressed in terms of the Laplace transform of the hitting probability for an $ (N-1) $- FAR system. Taking its inverse Laplace transform gives the corresponding hitting probability in the time domain, as shown in Theorem \ref{threc}.
\begin{theorem}\label{threc}
	In an MC system with $N- $FARs, the  probability that an IM emitted from the point source at the origin hits  FAR$ _1 $  within time $ t $ is given as 
	\begin{align}
	\ehp{}\left(t,\x_1\mid \{\x_j\}_{j=2}^N\right)\approx\hp{}\left(t,\x_1\mid \{\x_j\}_{j=2}^N\right)=\mathcal{L}^{-1}\left[\lhp{}\left(s,\x_1\mid \{\x_j\}_{j=2}^N\right)\right],\label{receq}
	\end{align}
where the Laplace transform $\lhp{}\left(s,\x_1\mid \{\x_j\}_{j=2}^N\right)$ of the hitting probability (fraction of IMs absorbed) in $N-$FAR system can be recursively written in terms of  the Laplace transform  of the hitting probability in $(N-1)$- FAR subsystems $\{x_j\}_{j=1}^{N-1}$ and $\{x_j\}_{j=2}^{N}$ as 
	\begin{align}
		&\lhp{}\left(s,\x_1{\mid} \{\x_j\}_{j=2}^N\right)=\frac{\lhp{}\left(s,\x_1\mid \{\x_j\}_{j=2}^{N-1}\right){-} s\lhp{}\left(s,\x_N{\mid} \{\x_j\}_{j=2}^{N-1}\right)\lhp{}\left(s,\x_1-\y_N{\mid} {\{\x_j\}_{j=2}^{N-1}-\y_N}\right)}{
		1{-}s^2\lhp{}\left(s,\x_N{-}\y_1{\mid} {\{\x_j\}_{j=2}^{N-1}{-}\y_1}\right)\lhp{}\left(s,\x_1{-}\y_N{\mid} {\{\x_j\}_{j=2}^{N-1}{-}\y_N}\right)}.\label{rexnresub}
	\end{align}
Here, $ \lhp{}\left(s,\x_N{-}\y_1{\mid} \{\x_j \}_{j=2}^{N-1}{-}\y_1\right)$ represents the Laplace transform of the  probability that an IM with its starting point  $ \y_1 $ at FAR$ _1 $ hits FAR$ _N $ when the medium consists of $N-1$ FARs at $\x_2,\x_3,\cdots,\x_N$. 
\end{theorem}
\begin{proof}
	See Appendix \ref{arec}.
\end{proof}
To demonstrate the applicability of the above two methods: direct method via Theorem \ref{tnfar} and recursive method via Theorem \ref{threc}, we present the hitting probability for a symmetric system with $ 4 $ FARs in the following corollary. This scenario is equivalent to placing the center of FARs at the vertex of a tetrahedron, and the transmitter is placed at the centroid of the tetrahedron.
\begin{coro}\label{cnrx1}
		In an MC system with $N=4$ FARs that are equidistant from the transmitter at the origin and also equidistant from each other (\ie $ r_i=r \text { and } R_{ij}=R,\forall\ i,j\in{1,2,3,4},\ i\neq j$), the hitting probability (fraction of IMs absorbed) of the IM on each of the FAR is given by
			\begin{align}
			\hp{1}(t)=\frac{a}{r}\sum_{n=0}^{\infty}\frac{(-3a)^n}{R^n}\erfc\left(\frac{r-a+n(R-a)}{\sqrt{4Dt}}\right).\label{eqnrxeq}
		\end{align}
	\begin{proof}
		See Appendix \ref{anrx1} for proof via both approaches.
	\end{proof}
\end{coro}
\begin{figure}[h]
	\centering
	\includegraphics[width=0.5\linewidth]{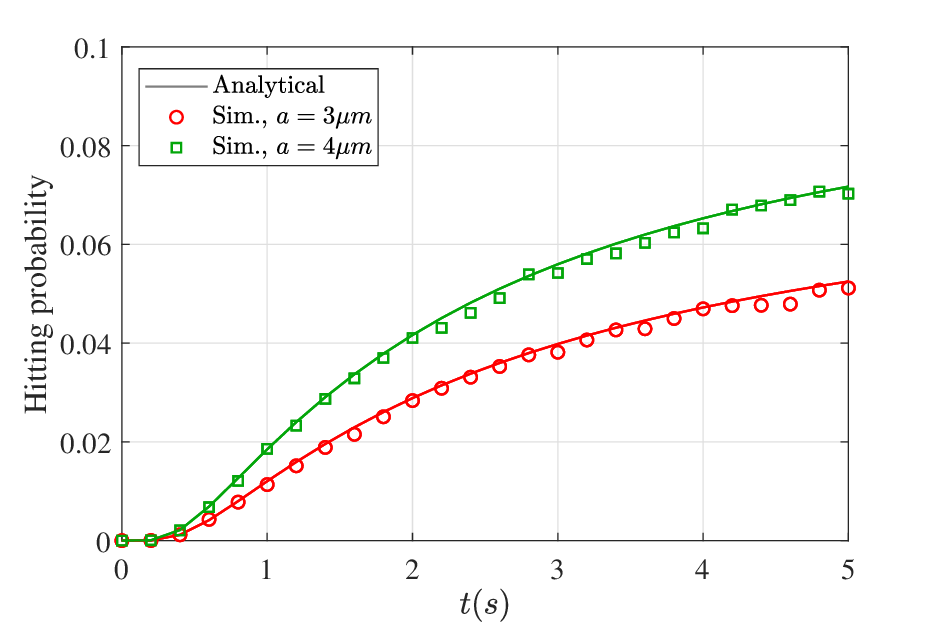}
	\caption{\color{black}Variation of the hitting probability until time $ t $ versus $ t $. Here, FARs are equidistant from the transmitter and equidistant from each other. Parameters: $ D=100\mu m^2/s $,  $ \x_1=[20,\ 20,\ 20],\ \x_2=[20,\ -20,\ -20],\ \x_3=[-20,\ -20,\ 20],\ \x_3=[-20,\ 20,\ -20] $.}
	\label{fig:coro24m1}
\end{figure}
   {\color{black} Fig. \ref{fig:coro24m1} validates the Corollary \ref{cnrx1} by comparing it with simulation.}  
    
  To demonstrate the importance and applicability of the derived channel models, we now describe a  system with a receiver consisting of multiple FARs arranged in a UCA fashion \cite{gursoy2019,tang2021,ahuja2021,gursoy2019a}.  We first derive the hitting probability expressions for this SIMO system in the next section. Using these  derived  expressions, we  analyze the performance of this system   in terms of the probability of bit error in Section \ref{sec:comm_UCA}. 

    \section{Hitting Probability in a SIMO System with UCA of FARs}\label{sec5uca}
        \begin{figure}[ht!]
    	\centering
    	\includegraphics[width=\figsize]{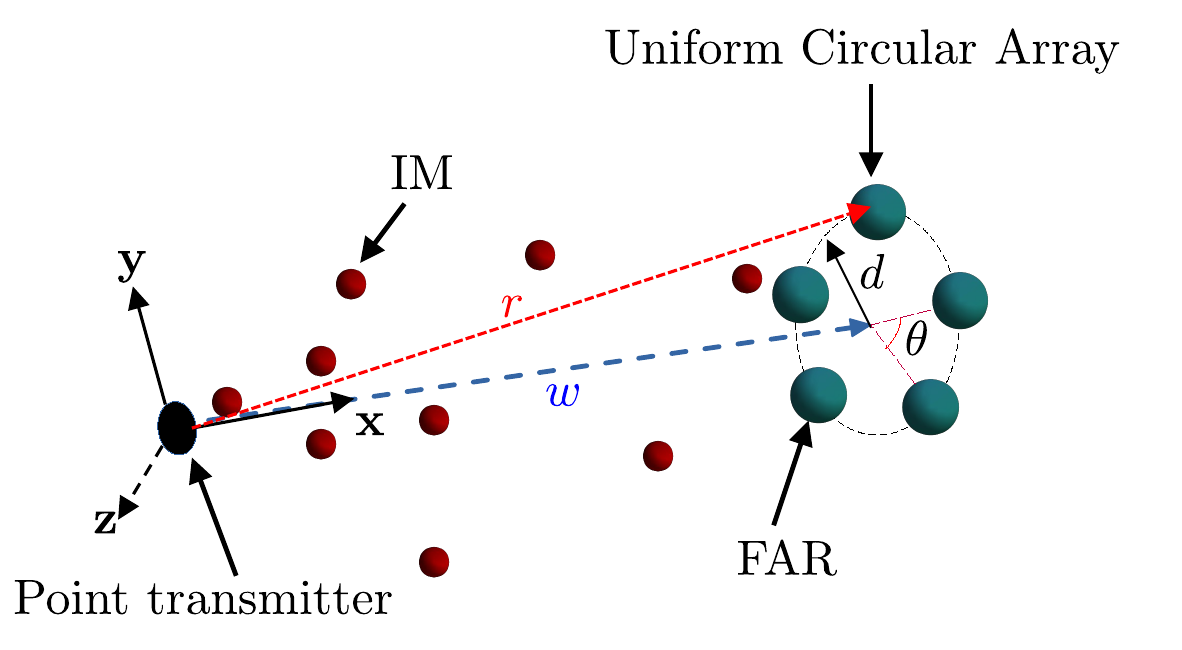}
    	\caption{An illustration showing a SIMO system with a point transmitter and a receiver consisting of $N $ FARs arranged as a UCA.}
    	\label{fig:uca}
    \end{figure}
    
    We now study a SIMO system with a point transmitter at the origin and multiple FARs arranged as UCA,  as shown in Fig. \ref{fig:uca}. The UCA center is located at $ x -$ axis at $ [w,0,0] $ and UCA is in $ y-z $ plane. The location of FAR$ _i $ is at $\x_i= [w,\ d\cos(2\pi n/N),\ d\sin(2\pi n/N)],\ n\in\{1,2,\cdots, N\} $.
The distance between the transmitter and each FAR is $r=\sqrt{w^2+d^2}$.
Without loss of generality, denote any one of the FARs in the UCA by FAR$_1 $ and the rest of the FARs by FAR$ _2 $ to FAR$ _N $ in the clockwise direction. 
From the perspective of the  distance from FAR$ _1 $ to the rest of the FARs, 
there are $\delta= \text{ceil}((N-1)/2) $ types of neighbors of FAR$_1$.  Let $ R_1 $ is the first neighbor distance, which is equal to $ R_{21} $ and $ R_{N1} $, $ R_2 $ is the second neighbor distance, which is equal to $ R_{31} $ and $ R_{(N-1)1} $ and so on.

The following theorem presents the hitting probability of an IM in a SIMO system with $ N $ FARs forming a UCA  and a single transmitter. 
\begin{theorem}\label{exUCA}
	In a SIMO system with $N$ FARs arranged as UCA, the  probability that an IM emitted from the transmitter hits FAR$ _1 $  is given as
	\begin{align}
		\ehp{1}(t)&=\invlaplace{}{
			\frac{\lhpone(s,r)}{1+\sum_{j=2}^{N}s\mathbb{E}_{\z_j}\left[\lhpone(s,\|\z_j-\x_i\|)\right]} }{}\label{exlapb1}.
	\end{align}
\begin{proof}
    For UCA, $ r_i=r $ and $ \ehp{i}(t)=\ehp{j}(t),\ \forall i, j $. Substituting these values in \eqref{exsnxr},  simplifying further and applying the inverse Laplace transform, we get \eqref{exlapb1}.
\end{proof}
\end{theorem}

    \begin{theorem}\label{NUCA}
    	In a SIMO system with $N$ FARs arranged as UCA, the probability that an IM emitted by the point source hits FAR$ _1 $  is approximately equal to\\
    	For odd N:
    	\begin{align}
    		\hp{1}(t)&=\sum_{n=0}^{\infty}(-2)^n\frac{a^{n+1}}{r}\sum_{\mid \mathbf{k}\mid=n} {n \choose \mathbf{k}}\frac{1}{\mathbf{R^k}}
    		\erfc\left(\frac{r-a+\sum_{m=1}^{\delta}k_m(R_m-a)}{\sqrt{4Dt}}\right),\label{NUCAO}
    	\end{align}
    	For even N:
    	\begin{align}	
    		\hp{1}(t)&=\sum_{n=0}^{\infty}(-1)^n\frac{a^{n+1}}{r}\sum_{\mid \mathbf{k}\mid=n} {n \choose \mathbf{k}}\frac{2^{n-k_\delta}}{\mathbf{R^k}}
    		\erfc\left(\frac{r-a+\sum_{m=1}^{\delta}k_m(R_m-a)}{\sqrt{4Dt}}\right).\label{NUCAE}
    	\end{align}
    	Here, $ \mathbf{k}=(k_1, k_2, \ldots, k_\delta) $, $ |\mathbf{k}|=k_1+k_2+\cdots +k_\delta $, ${n \choose \mathbf{k}}={n \choose k_1, k_2, \ldots, k_\delta}  $
    	and $ \mathbf{R^k}=\prod_{i=1}^\delta (R_i)^{k_i}$ based on multi-index notation \cite{bricogne2010}. Also, the second sum is taken over all combinations of non-negative integer indices $ k_1 $ through $ k_\delta $ such that $ \sum_{i=1}^\delta k_i=n $. 
    	\begin{proof}
    		See Appendix \ref{ConjPr}.
    	\end{proof}
    \end{theorem}
The following corollary gives the upper and lower bounds of the hitting probability.
	\begin{coro}\label{boundcc}
			The lower and upper bound of the exact hitting probability of each FAR in a UCA with $N$ FARs is given as  \eqref{NUCAO} and \eqref{NUCAE} with $ R_m$ for all  $m$ replaced with $ \underline{R}_m=\|\x_m-\x_1\|-a $ and $ \overline{R}_m=\|\x_m-\x_1\|+a $ respectively.
	\end{coro}
	\begin{proof}
		Using the inequalities  $ \mathbb{E}_{\z_j}\left[\lhpone(s,\|\z_j-\x_i\|) \right]\leq \lhpone(s,\|\x_j-\x_i\|-a)  $ and $ \mathbb{E}_{\z_j}\left[\lhpone(s,\|\z_j-\x_i\|) \right] \geq  \lhpone(s,\|\x_j-\x_i\|+a)  $ in \eqref{exlapb1} and solving in a  similar fashion as the proof of Theorem \ref{NUCA} gives the desired bounds respectively.
\end{proof}

  The hitting probability presented in Theorem \ref{NUCA} can be simplified for specific values of $ N $. For example, for $ N=2 $ to $ 5 $, the hitting probability expressions are given in the following corollary.
\begin{coro}\label{tuca}
    	In a SIMO system with $N$ receiver FARs arranged as a UCA, the  probability that an IM emitted from the point source at the origin hits a FAR is given as
    	\begin{align}
    		\hp{1}(t)&=\frac{a}{r}\sum_{n=0}^{\infty}\frac{(-a)^n}{R_1^n}\erfc\left(\frac{r-a+n(R_1-a)}{\sqrt{4Dt}}\right),\label{UCA2}&&\text{if } N=2\\
    		\hp{1}(t)&=\frac{a}{r}\sum_{n=0}^{\infty}\frac{(-2a)^n}{R_1^n}\erfc\left(\frac{r-a+n(R_1-a)}{\sqrt{4Dt}}\right),\label{UCA3}&&\text{if } N=3\\
    		\hp{1}(t)&=\sum_{n=0}^{\infty}(-1)^n\frac{a^{n+1}}{r}\sum_{k=0}^{n}{n \choose k}\frac{2^k}{R_1^kR_2^{n-k}}&&\nonumber\\ &\times\erfc\left(\frac{r-a+k(R_1{-}a)+(n-k)(R_2{-}a)}{\sqrt{4Dt}}\right),&&\text{if } N=4\label{UCA4}\\
    		\text{and }
    		\hp{1}(t)&=\sum_{n=0}^{\infty}(-2)^n\frac{a^{n+1}}{r}\sum_{k=0}^{n}{n \choose k}\frac{1}{R_1^kR_2^{n-k}}\nonumber&&\\ &\times\erfc\left(\frac{r-a+k(R_1{-}a)+(n-k)(R_2{-}a)}{\sqrt{4Dt}}\right),&&\text{if } N=5\label{UCA5}
    	\end{align}
    	respectively.
    \end{coro}

\subsection{Numerical Evaluation}\label{nsec}
\begin{figure}
	\centering
	\includegraphics[width=0.55\linewidth]{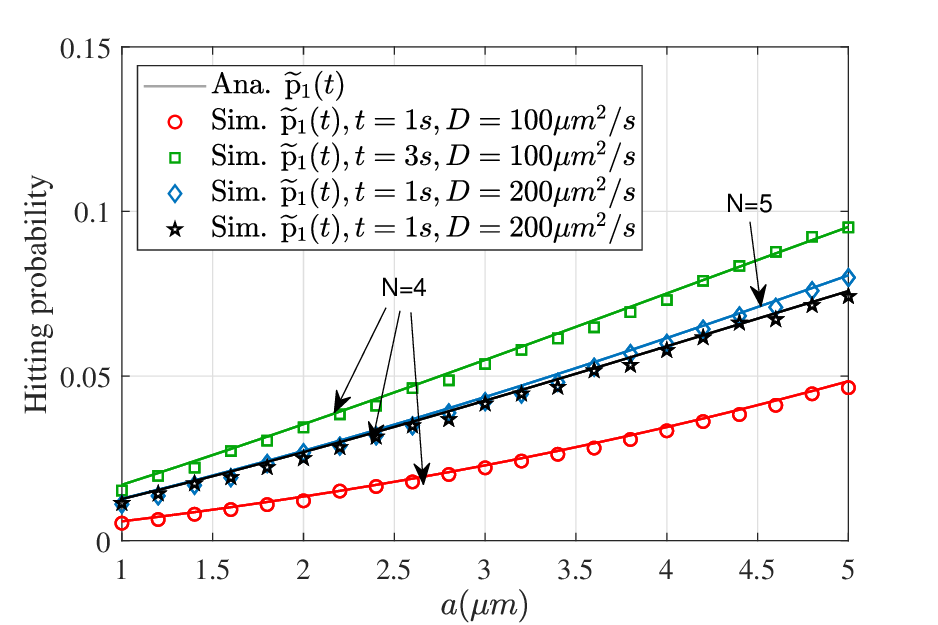}
	\caption{Variation of the hitting probability $\hp{1}(t)$ with the radius $a$ of the FAR for different values of $D$ and $ t $. Here, $d=20 \mu m, w=10\mu m,\ r=22.36\mu m,\
		\Delta t=10^{-4}\ s $.}
	\label{fig:ucaf3}
\end{figure}

The variation of the hitting probability with the radius $a$ of the FARs for $ N=4,\, 5 $, $ t=1\text{ and } 3s $ and $ D=100,\ 200\mu m^2/s $ is shown in Fig. \ref{fig:ucaf3}. We can verify that the approximate analytical results given in \eqref{UCA4} and \eqref{UCA5} are in good match with the particle-based simulation results. {\color{black}The number of terms to be considered for the hitting probability expression depends on the parameters of the system and the required level of accuracy. In this paper, different configurations and parameters of MC systems are considered for numerical analysis. Therefore, we chose the following alternative method for simulation. Let $ \epsilon $ be the difference between the summation of the hitting probability expression with $ n-1 $ terms and the $ n $th term. For simulation, we found that $\epsilon=10^{-10}  $ can provide the required  accuracy, which corresponds to $ n=6 $ to $ 31 $ for $ N=2 $ to $ 11 $.}
We can observe that the hitting probability increases with the radius of the FARs. This is because an increase in $a$ causes an increase in the surface area of the FARs, which allows more IMs to hit the FARs. The widening gap between the hitting probabilities for $ N=4 $ and $ N=5 $ with $ a $ is due to the increase in the mutual influence by adding one extra FAR and the reduction in the mutual distance between adjacent FARs. 

\begin{figure}[h]
	\centering
	\includegraphics[width=0.55\linewidth]{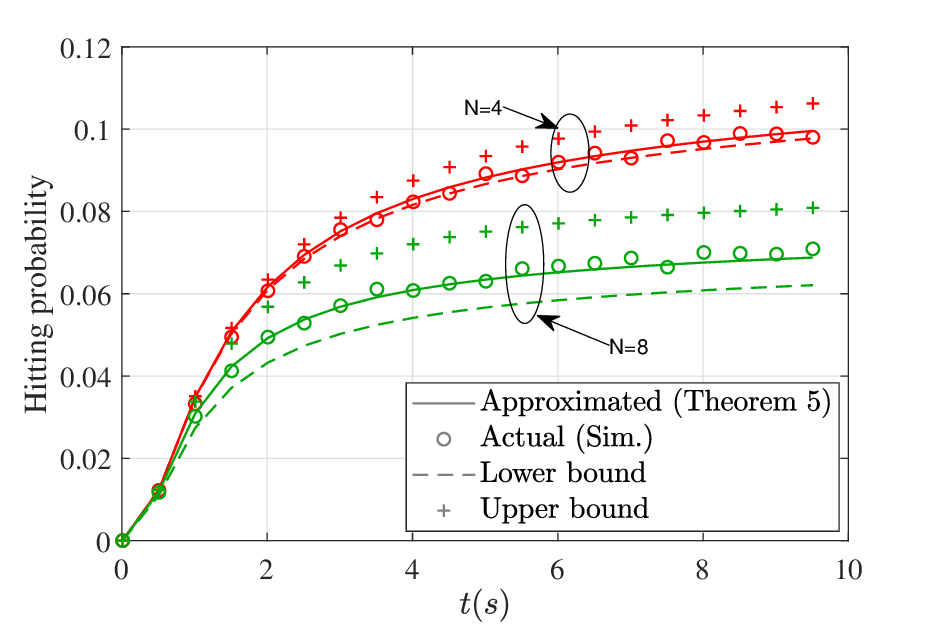}
	\caption{\color{black}Variation of the actual, approximate, lower and upper bounds of hitting probability with time for a UCA of $ N $ FARs. Parameters: $ a = 4\mu m,\ D = 100\mu m^2/s,\ d = 20\mu m,\ w = 10\mu m,\ t = 10^{-4}s $.}
	\label{fig:ucaboundm}
\end{figure}
{\color{black}From Fig  \ref{fig:ucaboundm}, we can verify that the derived analytical results (represented by bold solid lines) and particle-based simulation results (represented by markers) are upper-bound and lower-bound as given in Corollary \ref{boundcc}.}

  \begin{figure}[ht!]
	\centering
	\includegraphics[width=0.3\linewidth]{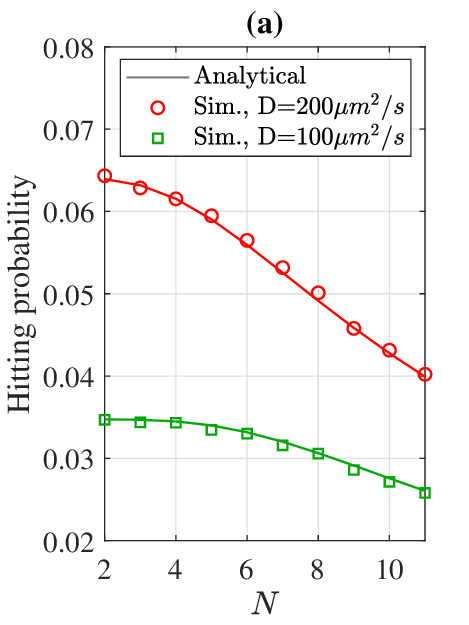}
	\includegraphics[width=0.3\linewidth]{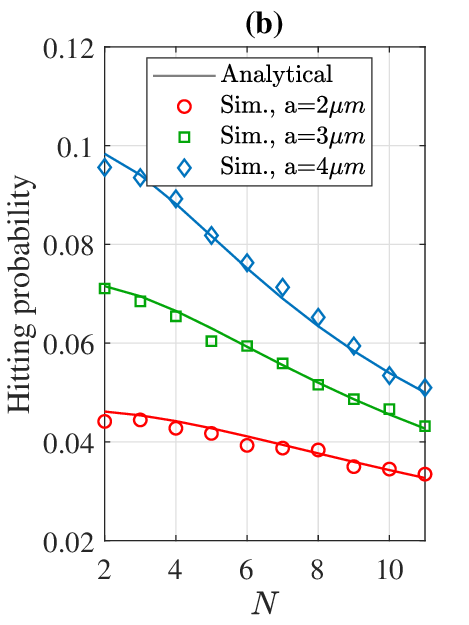}
	\caption{Variation of the hitting probability $\hp{1}(t)$ with $ N $ for different values of (a) $ D $ at $ t=1s $ and (b) $ a $ at $ t=5s $ for a SIMO system with $N$ FAR UCA. Here, $d=20 \mu m, w=10\mu m,\ r=22.36\mu m,\ a=4\mu m,\
		\Delta t=10^{-4}\ s$.}
	\label{fig:nuca}
\end{figure}

 Fig. \ref{fig:nuca} shows the variation of the hitting probability with the number of FARs. We can observe that the hitting probability drops with $N$ due to the increase in the mutual influence of FARs.
When $D$ increases, the hitting probability increases due to the faster motion of the molecules. An increase in  $ D $ causes molecules to wander more, causing them to bump to other FARs with a higher chance. This can increase the mutual influence of FARs. The hitting probability also decreases faster with an increase in $ N $ for higher values of $a$. This is due to the increase in the mutual influence of FARs with higher surface areas.

\subsection{Signal Gain in the Received IMs due to Multiple FARs}\label{siggain}
Now, we study the gain in the received signal due to the use of $ N $ FARs. When only one FAR is present, the received signal is equal to the FAR's hitting probability multiplied with the number of emitted molecules $ M $ \ie $ M\hpone_i(t) $.  
Assuming a symmetric  combining of the individual received signals of each FAR, the total received signal when $ N $ FAR is used is given as $ \sum_{i=1}^NM\hp{i}(t) =MN\hp{i}(t) $. Hence, the {\em signal gain} $ g(t) $ is given as
\begin{align}
	g(t) =\frac{\text{Average number of IMs absorbed by an $ N $ FAR system}}{\text{Average number of IMs absorbed by a single FAR system}} =\frac{N\hp{i}(t)}{\hpone(t,r)}
	,\label{eqgain}
\end{align}
which means that $ g(t) $ is the same as the ratio of the summation of hitting probabilities of IMs on each of the $ N $ FARs to the hitting probability for a single FAR case. {\color{black}Note that the transmitter to FAR distance is fixed as $ r $ for a fair comparison of multiple FAR case and a single FAR case.}

\begin{coro}\label{thdg2rx}
	The signal gain $ g(t) $ in a SIMO system with UCA of $ 2 $- FAR and $ 3 $- FAR  is respectively given as
	\begin{align}
		g(t)&=\frac{N\hp{1}(t)}{\hpone(t,r)}=
		\begin{cases}
		\displaystyle 2-\frac{2\lambda(t,r,1)}{\hpone(t,r)},& \text{if } N=2
	\\
	\displaystyle 3-\frac{3\lambda(t,r,2)}{\hpone(t,r)},& \text{if } N=3
	\end{cases}\\
	\label{dg3}
	\text{where }	\lambda(t,r,b)&=\frac{a}{r}\sum_{n=1}^{\infty}\frac{(-1)^{n+1}(ba)^n}{R^n}\erfc\left(\frac{r-a+n(R-a)}{\sqrt{4Dt}}\right).
	\end{align}
	\begin{proof}
		See Appendix \ref{adg2rx}.
	\end{proof}
\end{coro}
Since $ \hp{i}(t)<\hpone(t,r) $ due to FAR's mutual influence, \eqref{eqgain} shows that the signal gain is less than $ N $, and it approaches $ N $ if mutual influence is small (for example, when $ R_{ij}\rightarrow\infty $). In particular, from Corollary \ref{thdg2rx}, we can verify that the gain $ g(t)\leq 2 $ and $ g(t)\leq 3 $ for $ N=2 $ and $ N=3 $ respectively. 
	Further, when all the FARs are far apart (\ie $ R_{ij}\rightarrow\infty,\forall \ i,j, i\neq j$), then $\lambda(t,r_1,b) \rightarrow 0 $. This implies that $ g(t)\rightarrow 2 $ and  $ g(t)\rightarrow 2 $ for $N=2$ and $N=3$ respectively.  
Since the total number of received molecules determines the bit error rate of the MC channel, the signal gain acts as a proxy metric for the performance of the MC system. A high value of signal gain indicates reliable communication. Therefore, the signal gain can be utilized to decide suitable FAR placement patterns and mutual distances while designing the system.
	
The following corollary presents the asymptotic signal gain  for UCA with $ N $ FARs.

\begin{coro}\label{cornrx}
The asymptotic signal gain for a SIMO system with $N$ FAR UCA 	at $ t\rightarrow\infty $ is  
	\begin{align}
		g_\infty= \lim_{t\rightarrow \infty}g(t)=\begin{cases}
			\displaystyle\frac{N}{1+\sum_{i=1}^{\delta}2a/R_i}& \text{for odd N} \\ 
			\displaystyle \frac{N}{1+\sum_{i=1}^{\delta-1}2a/R_i+a/R_\delta}& \text{for even N}.
		\end{cases}
	\end{align}
\end{coro}
\begin{proof}
	Application of the final value theorem for the Laplace transform (\ie $ \lim_{t\rightarrow\infty} f(t)=\lim_{s\rightarrow0}sF(s)$) \cite{spiegel1992} in  \eqref{lapucan}  gives the desired result. 
\end{proof}
The above results show that the maximum gain that can be achieved eventually using UCA of $ N $ FARs is always less than $ N $, unless FARs are located far away from each other (\ie $ R_i\rightarrow \infty $). 
\begin{figure}[ht!]
	\centering
	\begin{subfigure}[b]{0.49\textwidth}
		\centering
	\includegraphics[width=0.9\textwidth]{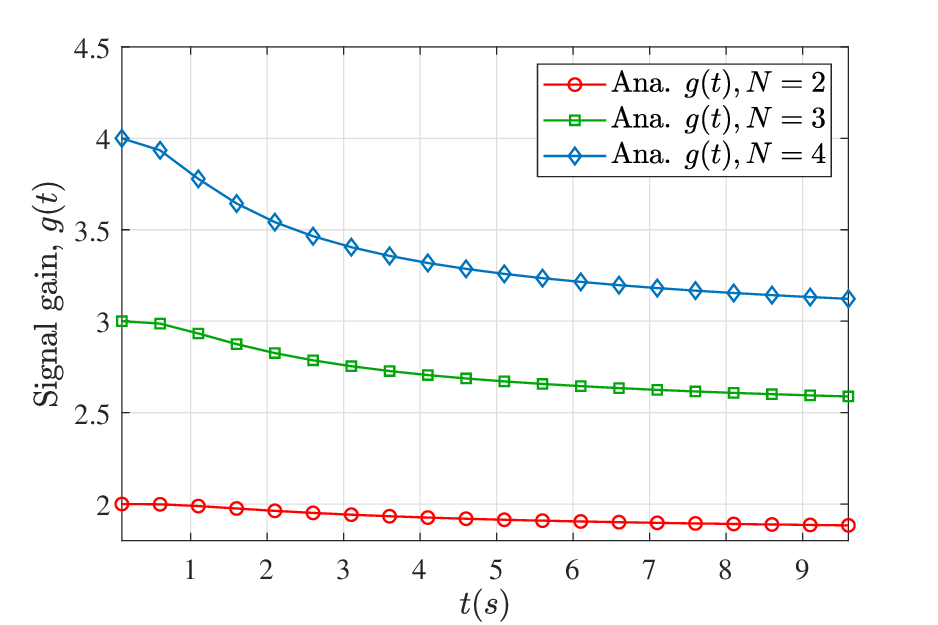}
	\caption{}
\end{subfigure}
\begin{subfigure}[b]{0.49\textwidth}
	\centering
	\includegraphics[width=0.9\textwidth]{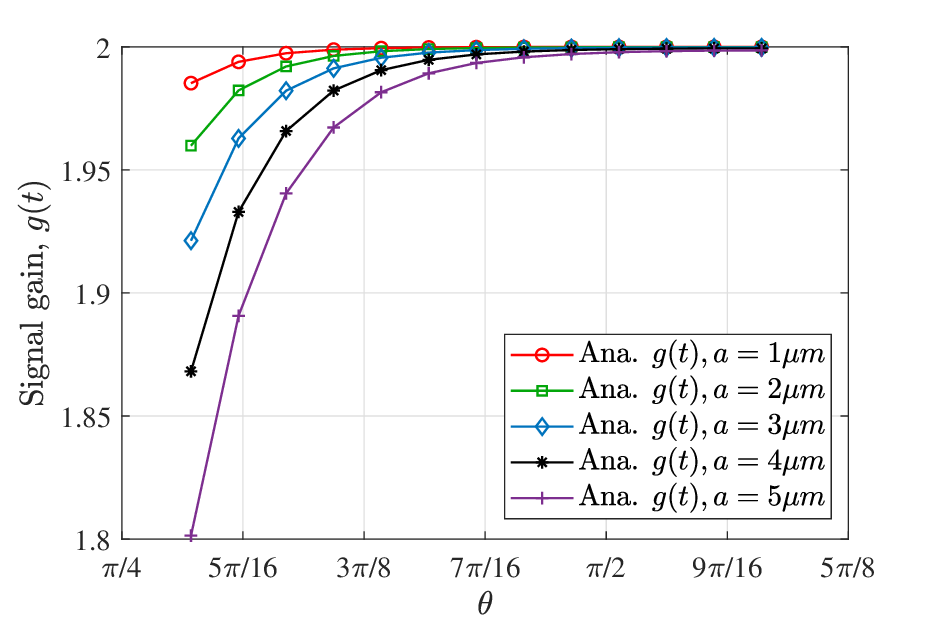}
	\caption{}
\end{subfigure}
	\caption{(a)Variation of signal gain with (a) time for different values of $ N $ in a SIMO system with $N$ FARs UCA. 
		Here, $d=20\mu m, r=20 \mu m, w=0\mu m,\ a=4 \mu m,\ D=100\mu m^2/s
		$, (b)  angular distance between two FARs for different values of $ a $ for FARs arranged in a circular array ($N=2$). Here, the two FARs are equidistant from the origin and, $\ D=100\mu m^2/s,\ t=1 s,\ \x_1=[d,0,0],\ \x_2=[d\cos(\theta),d\sin(\theta),0] ,\ d=20\mu m$.}
	\label{fig:dgt}
	
\end{figure}

Fig. \ref{fig:dgt} (a) shows the variation of $ g(t) $ with $ t $ for different values of $ N $ in a system for  UCA of $ N $ FARs.
When $ t $ increases, the $ g(t) $ reduces from $ g(t)\approx N $ to a lower value. The reduction in the $ g(t) $ with $ t $ is due to the increase in the mutual influence between FARs with time as seen in Fig. \ref{fig:abserr2m1}. 
 When $ t\rightarrow \infty $, from Corollary \ref{cornrx}, 
$g_\infty$ is given by $ g_\infty=\frac{2}{1+(a=4)/(R_1=36)}=1.8,\ g_\infty=\frac{3}{1+2(a=4)/(R_1=31.24)}=2.39 $ and $ g_\infty=\frac{4}{1+2(a=4)/(R_1=25.61)+(a=4)/(R_2=36)} =2.81$, respectively for $ N=2,3 $ and $ 4 $.


Fig. \ref{fig:dgt} (b) shows the variation of the signal gain with the angle between two FARs for different values of $ a $ in a circular array arranged $ 2-$FAR system.
The increase in $ g(t) $ with the increase in the angle between the two FARs is due to the decrease in mutual influence between the FARs. Also, we can see that an increase in $ a $ increases the $ g(t) $. The result further shows that with an increase in $N$, the angular distance between the FARs decreases, causing the mutual influence to decrease the signal gain significantly.

   Now, we use the derived hitting probability expression for the UCA arrangement of FARs to analyze the performance of the SIMO system with different cooperative detection schemes in the next section.

    \section{Communication Performance in a SIMO System with UCA of $N$ Co-operative FARs}\label{sec:comm_UCA}
    
Let us consider a SIMO communication system consisting of a point transmitter and a receiver consisting of a UCA of $N$ FARs as described in the previous section.
Let the duration of each time slot be $ \tb $. The point transmitter impulsively emits molecules at the beginning of the time slot corresponding to the transmit information bit  $ b[k] $ at that time slot. 
$ b[k] $ is Bernoulli distributed with parameter $ q $. We consider on-off keying (OOK) based modulation, in which the transmitter emits $ u[k]=M $ number of molecules for bit $ b[k]=1 $ and no molecules for bit $ b[k]=0 $.  The fraction of IMs transmitted at the time slot $ k $ that reach FAR$ _i$ in the slot $ l $ is given by
 \begin{align}
   	h_{i}[l-k]=&\hp{i}((l{-}k{+}1)\tb)-\hp{i}((l{-}k)\tb).\label{simoccoeff}
   \end{align}  
Let $ z_{i}[k;l] $ denote the number of IMs received at FAR$ _i $ at the $ l $th time slot due to 
the $ k $th time slot transmission. The indicator that an IM emitted at the $ k $th slot hits the FAR$ _i $ in the $ l $th slot is Bernoulli random variable with parameter $ h_{i}[l-k]  $. Therefore, $z_{i}[k;l] $ is Binomial distributed with parameters $ b[k]u[k] $ and $ h_{i}[l-k] $, \ie $ z_{i}[k;l]\sim\mathcal{B}( b[k]u[k],h_{i}[l-k] ) $. For mathematical tractability,  
the distribution of $ z_{i}[k;l] $ can be approximated as the  Poisson distribution having mean $ b[k]u[k]h_{i}[l-k] $ (\ie $  z_{i}[k;l]\sim\mathcal{P}\left( b[k]u[k]h_{i}[l-k]\right)$) \cite{jamali2019}. The total number of molecules received at the FAR$_i $ in the $ l $th time slot (\ie $ z_{i}[l] $) is given as 
   \begin{align}
   	z_{i}[l]=z_{i}[l;l]+\sum_{k=1}^{l-1} z_{i}[k;l].\label{simorxmolwith}
   \end{align}
   where $z_{i}[l;l]  $) denotes the number of signal molecules and $ \sum_{k=1}^{l-1} z_{i}[k;l]$ denotes the inter-symbol interference (ISI) molecules.
From \eqref{simorxmolwith}, conditioned on the transmit bit $b[l]$, the mean number of IMs received at FAR$ _i $ 
is given by
\begin{align}
	\lambda_{0i}[l]&=\expectA{z_{i}[l]|b[l]=0}=qM\sum_{k=1}^{l-1}h_i[l-k]
	=qM\left(\hp{i}(l\tb)-\hp{i}(\tb)\right)\label{emu0}\\
	\lambda_{1i}[l]&=\expectA{z_{i}[l]|b[l]=1}=M\hp{i}(\tb)+\lambda_{0i}[l]\label{emu1}.
	\end{align}%
\subsection{Cooperative Detection in SIMO Systems}
In MC, the FARs can cooperate by sharing their received signals $ z_{i}[l] $  to jointly decode. Depending on the information-sharing mechanism between the FARs, the following combining schemes can be utilized by the FARs.
	\subsubsection{Soft Combining}
If FARs can share complete information about their received signals, they can soft-combine these signals. Since the system is symmetric, they apply the same weight to each individual signal. Hence, the resultant received signal is the sum of individual received signals, which is equal to the total number of molecules received at any FAR given as 
		\begin{align}
			z[l]=\sum_{i=1}^{N} z_{i}[l].\label{simorxmolsum}
		\end{align}
		From \eqref{simorxmolsum}, conditioned on the transmit bit $b[l]$, the mean total number of IMs received is given by
		\begin{align}
			\lambda_{0}[l]&=\expectA{z[l]|b[l]=0}=qMN\left(\hp{1}(l\tb)-\hp{1}(\tb)\right)\label{emu00}\\
			\lambda_{1}[l]&=\expectA{z[l]|b[l]=1}=MN\hp{1}(\tb)+\lambda_{0}[l]\label{emu11}.
	\end{align}
	\subsubsection{Hard Combining}
	Alternatively, the FARs can individually make a \textit{local decision} based on their individual received signal $ z_i[l] $. These hard decisions can be combined to perform the final decoding. In particular, they can employ the rule to decode the transmitted bit as $ 1 $, if at least $ K $ of the $ N $ FAR’s local decision is $ 1 $, otherwise $ 0 $ \cite{yutingfang2017,varshney2018a}. Some special cases of $ K $-out-of-$ N $ fusion rule are
	\begin{itemize}
		\item \textit{Majority rule}: The transmitted bit is decoded as $ 1 $ if majority of the FARs decode it as $ 1 $, \ie  $ K=\floor{N/2}+1 $. 
		\item \textit{OR rule}: The transmitted bit is decoded as $ 1 $ if at-least one FAR decodes it as $ 1 $, \ie $ K=1 $.
		\item \textit{AND rule}: The transmitted bit is decoded as $ 1 $ if all FARs decode it as $ 1 $, \ie $ K=N $.
	\end{itemize}
Since less information is required to be shared among FARs compared to soft-combining, hard combining is more practical. However, due to the same reasons, its performance is lower than that of the soft-combining.

We consider perfect reporting of information from the FARs to the fusion center (FC) in both schemes. A fully transparent FC with perfect reporting of the absorbed number of IMs or the local decisions can be challenging to implement in real scenarios. Hence, the derived bit error probability expression can be used as the upper bound of any practical implementation.
\subsection{The Probability of Bit Error}
\subsubsection{Soft Combining}
 The combined signal $ z[l] $ is compared with the threshold for detection $ \eta[l]$.  The decision rule for the detection process can be represented as
\begin{align}
	z[l] \underset{\mathcal{H}_0[l]}{\overset{\mathcal{H}_1[l]}{\gtrless}}\eta[l],\label{hypsoft}
\end{align} 
where $ \mathcal{H}_{0}[l] $ and $ \mathcal{H}_{1}[l] $ are the null and alternate hypothesis at the FC corresponding to the bit value $ 0 $ and $ 1 $ respectively.
The overall probability of miss detection $ \mathrm{P}_{\m} [l]$ and probability of false alarm $ \mathrm{P}_{\f} [l]$ are given as
\begin{align}
	\mathrm{P}_{\m} [l]&=\mathrm{P}\left(z[l]<\eta[ l] \mid \mathcal{H}_{1}[l]\right)=\exp\left(-\lambda_{1 }[l]\right)\sum_{n=0}^{\eta[l]}\frac{\lambda_{1 }[l]^n}{n!},\label{spm}\\
	\text{and  }
	\mathrm{P}_{\f } [l]&=\mathrm{P}\left(z[l]\geq \eta[ l] \mid \mathcal{H}_{0}[l]\right)
	=1-\exp\left(-\lambda_{0}[l]\right)\sum_{n=0}^{\eta[l]}\frac{\lambda_{0}[l]^n}{n!}.\label{spf}
\end{align} 
Therefore, the probability of bit error is given as
\begin{align}
	\mathrm{P_e}[l]=q\mathrm{P_m} [l]+(1-q)\mathrm{P_f} [l].
\end{align}
\subsubsection{Hard Combining}
Since each FAR has to make its own hard decision, it compares the number of IMs received at itself with a threshold $ \eta[i,l] $. Hence, the local decision rule can be written as
	\begin{align}
		z_i[l] \underset{\mathcal{H}_{0i}[l]}{\overset{\mathcal{H}_{1i}[l]}{\gtrless}}\eta[i;l],\label{edetrule}
	\end{align} 
where $ \mathcal{H}_{0i}[l] $ and $ \mathcal{H}_{1i}[l] $ are null and alternate hypothesis of the $ i $th FAR.
The local probability of miss detection $ \mathrm{P}_{\m i} [l]$ and probability of false alarm $ \mathrm{P}_{\f i} [l]$
of FAR$_i  $ are given as
\begin{align}
	\mathrm{P}_{\m i} [l]&=\mathrm{P}\left(z_{i}[l]<\eta[i ; l] \mid \mathcal{H}_{1i}[l]\right)=\exp\left(-\lambda_{1 i}[l]\right)\sum_{n=0}^{\eta[i;l]}\frac{\lambda_{1 i}[l]^n}{n!},\label{lpm}\\
\text{and  }
	\mathrm{P}_{\f i} [l]&=\mathrm{P}\left(z_{i}[l]\geq \eta[i ; l] \mid \mathcal{H}_{0i}[l]\right)
	=1-\exp\left(-\lambda_{0 i}[l]\right)\sum_{n=0}^{\eta[i;l]}\frac{\lambda_{0 i}[l]^n}{n!}.\label{lpf}
\end{align}

Note that, due to the symmetric arrangement of FARs in the UCA, the distribution of the received number of molecules at each FAR is identically distributed. 
Therefore, the detection threshold (and hence, the local miss detection and false alarm probability) are equal across the FARs.
Hence, at the $ l $th time slot, the global miss detection and false alarm probabilities for the $ K $-out-of-$ N $ fusion rule are given by \cite{varshney1997,yutingfang2017}
\begin{align}
	\mathrm{Q}_{\m} [l]&=1-\sum_{k=K}^{N}\binom{N}{k}\left(1-\mathrm{P}_{\m} [l]\right)^k\mathrm{P}_{\m} [l]^{N-k},\label{gpm}\\
\text{and } 
	\mathrm{Q}_{\f} [l]&=\sum_{k=K}^{N}\binom{N}{k}\mathrm{P}_{\f} [l]^{k}\left(1-\mathrm{P}_{\f} [l]\right)^{N-k},\label{gpf}
\end{align}
respectively.
Therefore, the the probability of bit error is given as 
\begin{align}
	\mathrm{P_e}[l]=q\mathrm{Q_m} [l]+(1-q)\mathrm{Q_f} [l].
\end{align}
For OR and AND rule, the probability of bit error is respectively given as
\begin{align}
	\mathrm{P_e}[l]&=q\mathrm{P}_\m[l]^N +(1-q)\left(1-(1-\mathrm{P}_\f[l])^N\right),\\
	\text{and }
	\mathrm{P_e}[l]&=q\left(1-(1-\mathrm{P}_\m[l])^N\right) +(1-q)\mathrm{P}_\f[l]^N.
\end{align}

\subsection{Numerical Results 
}\label{copnum}
We now present numerical results for the performance of the SIMO system. Here, we have taken FAR radius $ a=4 \mu m$, diffusion coefficient $  D=100\mu m/s^2 $,  number of molecules $ M=200 $, symbol time $ \tb=5 s$, UCA center distance $ w=25 \mu m $ and UCA radius $ d=10\mu m $ unless stated otherwise. FARs are located at $ \x_i=[w,d\cos\left(\frac{2\pi (i-1)}{N}\right),d\sin\left(\frac{2\pi (i-1)}{N}\right)] $.
Here, the value of  $ l $ is fixed as $ 9  $ (\ie current time-slot is the $ 9 $th slot) because the average ISI molecules were observed to be almost  contained in $ 8 $ slots for the considered values of the parameters  $ a,\ D,\ M,\ \tb,\ r,\ w  $ and $ d $. Therefore, the BER will not change if we increase the value of $ l $ above $  9 $.
\begin{figure}[ht!]
	\centering
	\includegraphics[width=\figsize]{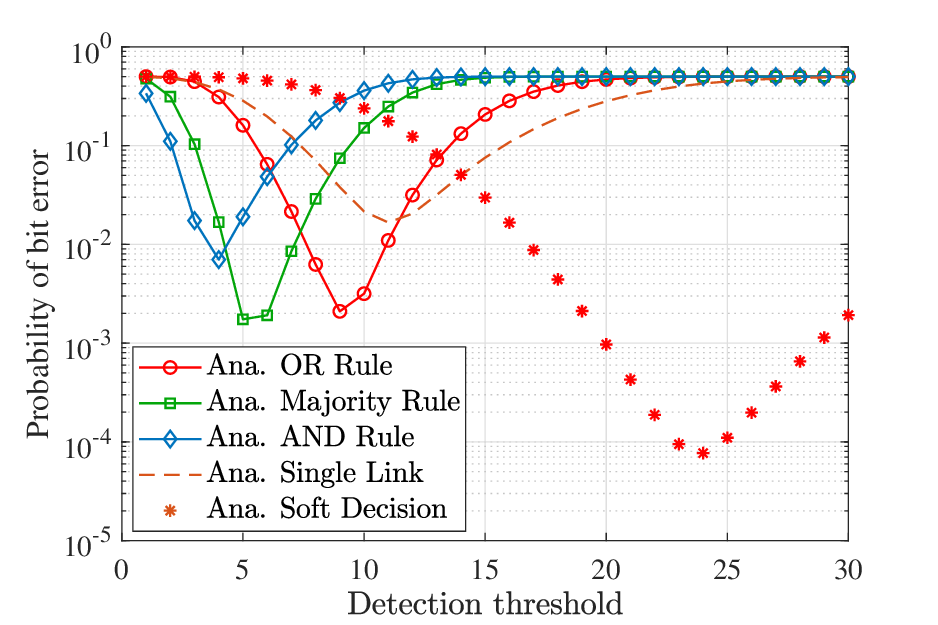}
	\caption{Variation of the probability of bit error in a SIMO MC system with UCA of FARs with the detection threshold. Here, $  N=4, \text{ and } q=0.5$.}
	\label{fig:pvsetauca}
\end{figure}

Fig. \ref{fig:pvsetauca} shows the variation of the probability of bit error  with the detection threshold for various combining methods. 
We can observe that, $\mathrm{P_e}[l]$ first reduces with the threshold. After reaching the minimum value, it increases, confirming the presence of an optimal detection threshold. For OR, AND, and majority rules, the optimal threshold ($ \eta[i;l] $) is near to  $ \lambda_{1i}[l] $, $ \lambda_{0i}[l] $ and in the middle of $ \lambda_{1i}[l] $ and $ \lambda_{0i}[l] $ respectively. The optimal threshold is different for different combining rules and has to be selected accordingly. The rest of the results use the optimal threshold value, which is computed numerically by varying the detection threshold for the specific set of parameter-values and combining rules to find the value for which $\mathrm{P_e}[l]$ is the minimum.

\begin{figure}[ht]
	\centering
		\begin{subfigure}[b]{0.49\textwidth}
		\centering
		\includegraphics[width=0.9\linewidth]{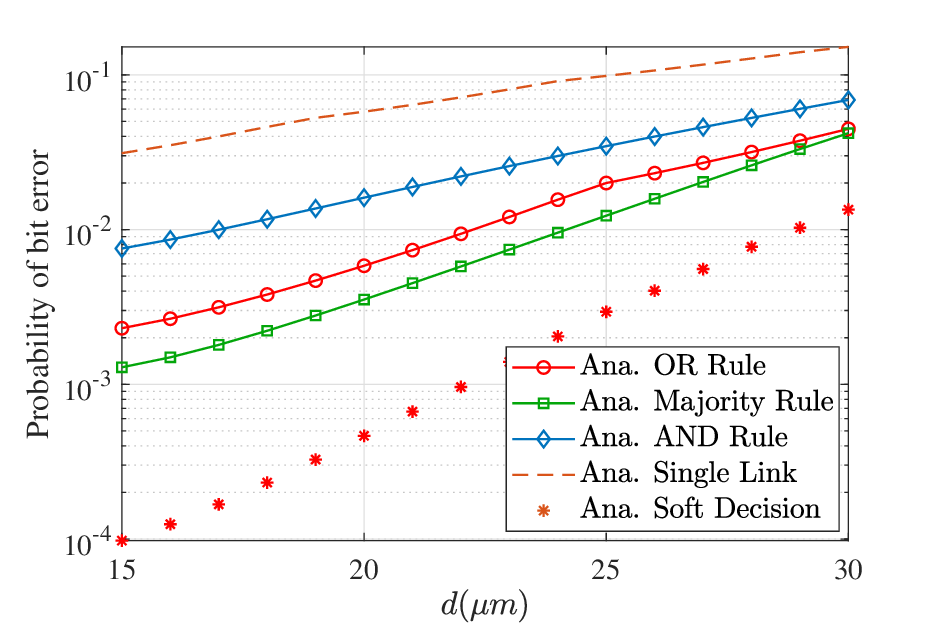}
		\caption{}
	\end{subfigure}
	\begin{subfigure}[b]{0.49\textwidth}
	\centering
	\includegraphics[width=0.9\linewidth]{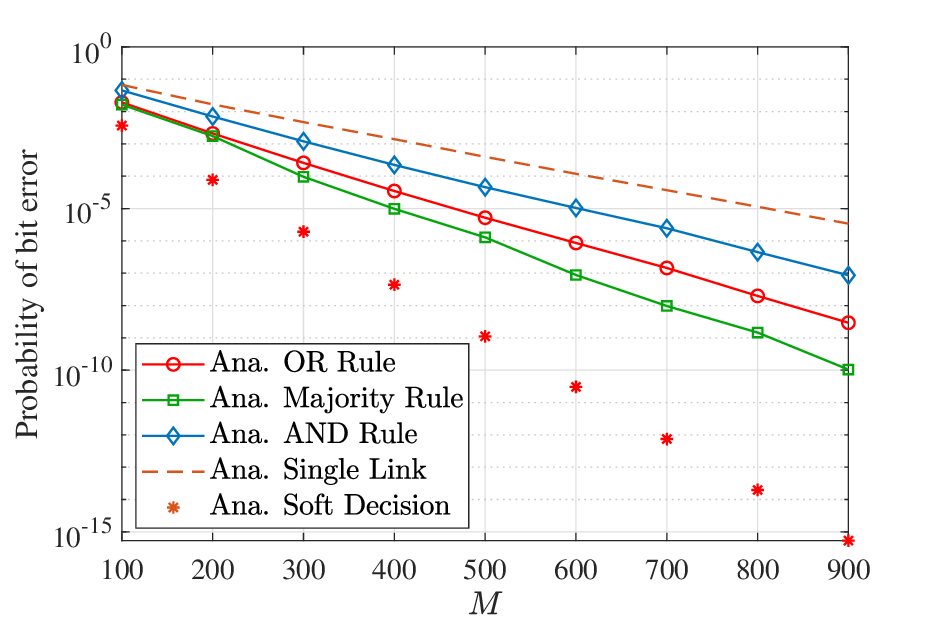}
	\caption{}
\end{subfigure}
	\caption{Variation of the probability of bit error with (a) the radius of UCA for a SIMO MC system with UCA, (b) the number of molecules emitted for bit $ 1 $. Parameters; $  N=4,\ q=0.5, \text{ and } \eta[l]=\text{optimal value}$.}
	\label{fig:pvsd}
\end{figure}

Fig. \ref{fig:pvsd} (a) shows the variation of $\mathrm{P_e}[l]$ with respect to  $ d $. An increase in $ d $ increases $ r $, thus reducing the hitting probability.  $\mathrm{P_e}[l]$ increases with $ d $.
Fig. \ref{fig:pvsd} (b) shows the variation of $\mathrm{P_e}[l]$ with respect to $ M $, which can be observed to decrease with $ M $.
 The effect of different decision rules becomes dominant when $ M $ is increased.
\begin{figure}[ht]
	\centering
	\includegraphics[width=\figsize]{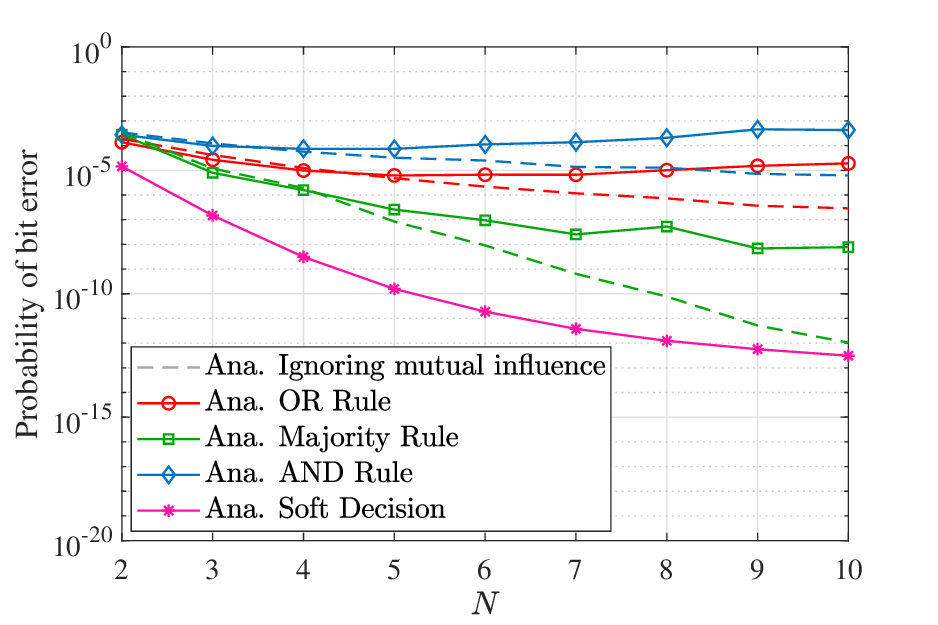}
	\caption{Variation of the probability of bit error with the number of FARs in a SIMO with a UCA receiver. Parameters; $d=15 \mu m, w=25\mu m,\ r=29.15\mu m,\ D=200\mu m^2/s,\ q=0.5,\ M=500,\ \eta[l]=\text{optimal value}$. The dashed lines shows the value ignoring the presence of other FARs. }
	\label{fig:pvsn1}
\end{figure}
Fig. \ref{fig:pvsn1} shows the variation of the probability of bit error with number $ N $ of FARs under various combining methods.
{\color{black} The dashed lines show $\mathrm{P_e}[l]$ for the case when other FAR's influence is ignored. This refers to the use of single FAR channel model expression \eqref{eq1rx} as the hitting probability of IM on each FAR in a system with multiple FARs.}

If mutual influence is ignored, the number of IMs captured per FAR remains constant with $N$. Therefore, the use of multiple FARs reduces decoding errors, and  $\mathrm{P_e}[l]$ is observed to reduce significantly with $N$. Due to mutual influence, the hitting probability at individual FARs decreases with $N$, and hence, the individual performance of each FAR degrades. Therefore, we observe a trade-off with $ N $ for hard combining rules. 
In soft combining, since the total number of IMs captured increases with $ N $, the probability of bit error always reduces with $ N $. We can verify that the bit error performance of the SIMO system is significantly better compared to a single input single output (SISO) system when complete information can be shared among the FARs.

Further, the mutual influence degrades the $\mathrm{P_e}[l]$ for a particular combining scheme. Therefore, it is essential to consider the effect of the mutual influence of FARs while designing the system. {\color{black}When channel models don't take into account the mutual influence between FARs, the number of received IMs in the FARs can be overestimated, resulting in an incorrect probability of bit error calculation.}

From the numerical results in Fig. \ref{fig:pvsetauca}-\ref{fig:pvsn1}  we can observe that soft combining is better than hard combining as expected. Furthermore, we observe that the majority rule performs better than OR, which performs better than the AND rule.

\section{Conclusions}
Molecular communication systems may consist of multiple receivers that are added to increase the communication systems' reliability and data rate. The exact analytical channel model for an MC system with multiple FARs does not exist in the literature due to the mathematical intractability. This work presents an approximate but simple hitting probability expression for the general $ N-$FAR case. The derived expressions are further simplified for various special cases, for example, when FARs are arranged in UCA and the transmitter is equidistant from the FARs. We also considered a SIMO system with a point transmitter and a receiver with UCA of $N$ FARs. Using the hitting probability expressions derived for the UCA case, we subsequently analyzed the error performance of cooperative detection in the considered SIMO system. We also compared the error performance under soft decision and hard decision rules (majority, OR, and AND rules). The derived hitting probability equations have many applications, including the analysis of interference and probability of bit error for  SIMO, MIMO, and cognitive systems.

\appendices
	\section{Derivation of Approximate Equation for the Fraction of IMs Absorbed at Each FAR in an $N-$FARs System}\label{anxr}
In \eqref{exsnxr}, we approximate the actual hitting point $ \z_j $ at FAR$_j$ by $ \y_j $, where $ \y_j $ is the closest point on FAR$ _j $ from the transmitter to get
\begin{align}
	\lhpone(s,r_i) -\lhp{i}(s)  =  \sum_{j=1,\ j\neq i}^{N} 
	s\lhp{j}(s) \lhpone(s,R_{ji}) , \ \ \forall\ i.\label{eqsnxr} 
\end{align}
Here $\lhp{i}(s)$ is the Laplace transforms of $\hp{i}(t)$.
The system of equations given in \eqref{eqsnxr} can be represented in matrix form as  $\Amat(s)\Pmat(s)=\bPmat(s)$, and the solution can be obtained as
\begin{align}
	\Pmat(s)=\Amat^{-1}(s)\bPmat_{}(s).\label{mateqn}
\end{align}
Finally, the hitting probability can be given by taking the inverse Laplace transform of both sides of \eqref{mateqn}. 
\section{Derivation of Lemma \ref{eboundc}}\label{ebound}
	Subtracting  \eqref{exsnxr} and \eqref{eqsnxr}, we obtain
	\begin{align}
		E_i(s)+\sum_{j=1,\ j\neq i}^{N} s\lhpone(s,R_{ji})E_j(s)=\sum_{j=1,\ j\neq i}^{N} 
		s\lehp{j}(s) \left( \lhpone(s,R_{ji})-\mathbb{E}_{\z_j}\left[\lhpone(s,\Rval_{ji})\right]\right)
		,\label{f2label}
	\end{align}
where $ E_i(s) =\lehp{i}(s)-\lhp{i}(s)$.
Now, stacking \eqref{f2label} for all $i$, we can write $ \mathbf{E}(s) =\mathbf{C}(s)\mathbf{X}(s)$, where,
$ \mathbf{E}(s)=\left[	E_1(s),E_2(s),\cdots, E_N(s)\right]$ , $ \mathbf{C}(s)= \Amat^{-1}(s)$, and $ \mathbf{X}(s)= \left[x_1(s),x_2(s),\cdots,x_N(s)\right]^{\mathrm{T}}$ with
\begin{align}
	x_k(s)=\sum_{j=1,\ j\neq k}^{N}s\lehp{j}(s) \left( \lhpone(s,R_{jk})-\mathbb{E}_{\z_j}\left[\lhpone(s,\Rval_{jk})\right]\right).\label{eqxx}
\end{align}
In particular, the absolute error for $i$th FAR is given by the absolute of the  $i$th element of $ \mathbf{E}(s) $ \ie
\begin{align}
	\left| E_i(s)\right| =\left| \sum_{k=1}^Nc_{ik}(s)x_k(s)\right| 
	\leq  \sum_{k=1}^N\left| c_{ik}(s)\right| \left| x_k(s)\right|,\label{eq:ev1}
\end{align}
where $c_{ik}$ denotes the $i,k$ element of $\mathbf{C}(s)$ and $x_k$ denotes the $k$th element of $\mathbf{X}(s)$.
Now, using the inequality $\lehp{j}(s)\leq \lhpone(s,r_j)  $ in \eqref{eqxx}, we get
\begin{align}
	\left| x_k(s)\right| 
	&\le\sum_{j=1,\ j\neq k}^{N}	s\lhpone(s,r_j)  F_{kj}(s),\label{new53}
\end{align}
where  $ F_{kj}(s) $ denotes the absolute error due to approximation of $ \z_j $ by $ \y_j $ defined as

$ 	F_{kj}(s)\overset{\Delta}{=}\mathbb{E}_{\z_j}\left[\left| \lhpone(s,\Rval_{jk}) -\lhpone(s,R_{jk}) \right|\right], \ \forall\ i $.

Now substituting \eqref{new53} in \eqref{eq:ev1} we get, 
\begin{align}
	\left| E_i(s)\right| 
	\leq  \sum_{k=1}^N\sum_{j=1,\ j\neq k}^{N}	\left| c_{ik}(s)\right|  F_{kj}(s) \ s\lhpone(s,r_j) .\label{erreq}
\end{align}
		Now, remaining is to compute  $ F_{ij}(s) $. Note that, 
		based on the hitting location $ \z_j $, $ \Rval_{ji}\leq R_{ji} $  or $ \Rval_{ji} >  R_{ji} $. \\
\textbf{Case 1:} When $ \Rval_{ji}\leq R_{ji} $ 
\begin{align}
	\mid\lhpone(s,\Rval_{ji})-\lhpone(s,R_{ji})\mid
	&=\stackrel{(a)}{\leq}\frac{a}{s}e^{a\sqrt{\frac{s}{D}}}e^{\left(-R_{ji}\sqrt{\frac{s}{D}}\right)}\left[\frac{e^{2a\sqrt{\frac{s}{D}}}}{\|\x_j-\x_i\|-a}-\frac{1}{R_{ji}}\right]=F_{ij}^{(1)}(s,z_j),\label{err1}
\end{align}
where $ (a) $ is due to the inequalities $ R_{ji}-\Rval_{ji}\leq \|\y_j-\z_j\|\leq 2a $, and $\Rval_{ji}\geq\|\x_j-\x_i\|-a  $. \\
\textbf{Case 2:} When $ R_{ji}\leq \Rval_{ji} $ 
\begin{align}
	\mid\lhpone(s,R_{ji})-\lhpone(s,\Rval_{ji})\mid
	&\stackrel{(b)}{\leq}\frac{a}{s}e^{a\sqrt{\frac{s}{D}}}e^{\left(-R_{ji}\sqrt{\frac{s}{D}}\right)}\left[\frac{e^{2a\sqrt{\frac{s}{D}}}}{R_{ji}}-\frac{1}{\|\x_j-\x_i\|+a}\right]=F_{ij}^{(2)}(s,z_j),\label{err2}
\end{align}
where $ (b) $ is due to the inequalities $ \Rval_{ji}-R_{ji}\leq \|\z_j-\y_j\|\leq 2a $, and $  R_{ji}\leq \Rval_{ji} $ and $\Rval_{ji}\leq\|\x_j-\x_i\|+a  $. \\
Therefore, $ F_{ij}(s)\leq \max\{F_{ij}^{(1)}(s,z_j),F_{ij}^{(2)}(s,z_j)\} $. Substituting the value of $ F_{ij}(s) $ in \eqref{erreq} gives Lemma \ref{eboundc}.
\section{Derivation of the Recursive Hitting Probability equation}\label{arec}
First notice that, for $ N=2 $, \cite[Eq. 19]{sabu2020a} and \cite[Eq. 20]{sabu2020a} state that
 \begin{align}
 	\lhpone{}(s,r_1)-\lhp{}(s,\x_1\mid \x_{2})=s\lhp{}(s,\x_2\mid \x_{1})\lhpone{}(s,R_{21}),\label{eq2rx1}\\
\text{and }
	\lhpone{}(s,r_2)-\lhp{}(s,\x_2\mid \x_{1})=s\lhp{}(s,\x_1\mid \x_{2})\lhpone{}(s,R_{12}).\label{eq2rx2}
\end{align}
Solving \eqref{eq2rx1} and \eqref{eq2rx2} gives
\begin{align}
	\lhp{}(s,\x_1\mid \x_{2})=\frac{\lhpone{}(s,r_1)-s\lhpone{}(s,r_2)\lhpone{}(s,R_{21})}{1-s^2\lhpone{}(s,R_{12})\lhpone{}(s,R_{21})}.\label{lp2rx}
\end{align}
 
Now, let us add one more FAR to have  a system of $ N=3 $ FARs. Here, the hitting probability of an IM that was supposed to hit FAR$ _{1} $ in the presence of FAR$ _2 $ within time $ t $ but is hitting FAR$ _{3} $ is given by 
\begin{align}
	&\hp{}(t,\x_1{\mid} \x_{2}) -\hp{}(t,\x_1{\mid} \{\x_j\}_{j=2}^3)   =  
	\int_{0}^{t} \frac{\partial\hp{}(\tau,\x_3{\mid} \{\x_j\}_{j=1}^2)}{\partial \tau}\hpone{}(t-\tau,\x_1{-}\y_3{\mid} \x_{2}{-}\y_3) \dd\tau , \label{req3rxeqn1}
\end{align}
Similarly, the hitting probability of an IM that was supposed to hit FAR$ _{3} $ in the presence of FAR$ _2 $ within time $ t $ but is hitting  FAR$ _{1} $ is given by 
\begin{align}
	&\hp{}(t,\x_3{\mid} \x_{2}) -\hp{}(t,\x_3{\mid} \{\x_j\}_{j=1}^2)   =   
	\int_{0}^{t} \frac{\partial \hp{}(\tau,\x_1{\mid} \{\x_j\}_{j=2}^3)}{\partial \tau}\hpone{}(t-\tau,\x_3{-}\y_1{\mid} \x_{2}{-}\y_1) 	\dd\tau   \label{req3rxeqn2}
\end{align}
 Now, taking Laplace transform on both sides of \eqref{req3rxeqn1} and \eqref{req3rxeqn2} gives
 \begin{align}
	\lhp{}(s,\x_1{\mid} \x_{2})-\lhp{}(s,\x_1{\mid} \{\x_j\}_{j=2}^3)=&s\lhp{}(s,\x_3{\mid} \{\x_j\}_{j=1}^2)
	\lhpone{}(s,\x_1{-}\y_3{\mid} \x_{2}{-}\y_3),\label{eq3rx1}\\
\text{and }
	\lhp{}(s,\x_3{\mid} \x_{2}){-}\lhp{}(s,\x_3{\mid} \{\x_j\}_{j=1}^2){=}&s\lhp{}(s,\x_1{\mid} \{\x_j\}_{j=2}^3)
	\lhpone{}(s,\x_3{-}\y_1{\mid} \x_{2}{-}\y_1).\label{eq3rx2}
\end{align}
Solving \eqref{eq3rx1} and \eqref{eq3rx2} gives
\begin{align}
	\lhp{}(s,\x_1{\mid} \{\x_j\}_{j=2}^3)=&
	\frac{
	\left(\lhp{}\left(s,\x_1\mid {\x}_{2}\right)\right.-\left. s\lhp{}\left(s,\x_3\mid {\x}_{2}\right)\lhp{}\left(s,\x_1-\y_3\mid {\x_{2}-\y_3}\right)\right)}
	{1{-}s^2\lhp{}\left(s,\x_3{-}\y_1{\mid} {\x_{2}{-}\y_1}\right)\lhp{}\left(s,\x_1{-}\y_3{\mid} {\x_{2}{-}\y_3}\right)}
	.\label{lp3rx}
\end{align}
The above approach can be extended for any value of $ N $ to get the following set of equations.
 \begin{align*}
	\lhp{}(s,\x_1{\mid} \{\x_j\}_{j=2}^{N-1})&-\lhp{}(s,\x_1{\mid} \{\x_j\}_{j=2}^N)=s\lhp{}(s,\x_{N}{\mid} \{\x_j\}_{j=1}^{N-1})
	\lhp{}(s,\x_1-\y_{N}{\mid} \{\x_j\}_{j=2}^{N-1}-\y_N),\\
\text{and }
	\lhp{}(s,\x_{N}{\mid }\{\x_j\}_{j=2}^{N-1})&-\lhp{}(s,\x_{N}{\mid} \{\x_j\}_{j=1}^{N-1})=s\lhp{}(s,\x_1{\mid} \{\x_j\}_{j=2}^N)
\lhp{}(s,\x_N-\y_{1}{\mid} \{\x_j\}_{j=2}^{N-1}-\y_1).
\end{align*}
Solving these equations simultaneously  gives \eqref{rexnresub}. Taking the inverse Laplace transform of \eqref{rexnresub} gives \eqref{receq}.

\section{Derivation of Corollary \ref{cnrx1}}\label{anrx1}
\textbf{Recursive method:}
On substituting $ r_i=r,\ \forall i  $ and $ R_{ij}=R,\ \forall i,j,\  i\neq j$ in \eqref{rexnresub} for $ N=4 $, owing to symmetry $ \lhp{}\left(s,\x_i\mid \{\x_j\}_{j=i,j\neq 1}^{N-1}\right){=} \lhp{}\left(s,\x_1\mid \{\x_j\}_{j=2}^{N-1}\right), \forall i$  and $ \lhp{}\left(s,\x_i{-}\y_j{\mid} {\{\x_j\}_{j=i,j\neq 1}^{j-1}{-}\y_j}\right){=}\\ \lhp{}\left(s,\x_1{-}\y_N{\mid} {\{\x_j\}_{j=2}^{N-1}{-}\y_N}\right), \forall i,j, \ i\neq j $. Note that the sub system of $ 3 $ FARs is also symmetric. Therefore, the common terms cancel out in \eqref{rexnresub}, and the Laplace transform of the hitting probability is given as
\begin{align}
	\lhp{}(s,\x_1\mid \{\x_j\}_{j=2}^4)=\frac{\lhp{}(s,\x_1\mid \{\x_j\}_{j=2}^3)}{1+s\lhp{}(s,\x_1-\y_4\mid \{\x_j\}_{j=2}^3-\y_4)},\label{tetra1}
\end{align}
where
\begin{align}
	\lhp{}(s,\x_1\mid \x_{2})&=\frac{\lhp{}(s,\x_1\mid \{\x_j\}_{j=2}^3)}{1+s\lhp{}(s,\x_1-\y_3\mid \x_{2}-\y_3)}
	=\frac{\lhpone{}(s,r)}{1+2s\lhpone{}(s,R)},\label{tetra2}
\end{align}
and 
\begin{align}
\lhp{}(s,\x_1-\y_4\mid \{\x_j\}_{j=2}^3-\y_4)=\frac{\lhpone{}(s,R)}{1+2s\lhpone{}(s,R)}.\label{tetra3}
\end{align}
Substituting \eqref{tetra2} and \eqref{tetra3} in \eqref{tetra1} gives
\begin{align}
	\lhp{}(s,\x_1\mid \{\x_j\}_{j=2}^4)&=\frac{\lhpone(s,r)}{1+3s\lhpone(s,R)}
	=\lhpone(s,r)\times\sum_{n=0}^{\infty}(-3)^n\left(s\lhpone(s,R)\right)^n.
\end{align}
Substituting the value of $ \lhpone(s,x),\ x\in\{r,R\} $ in the above equation using \eqref{lapeq} and taking its inverse Laplace transform gives Corollary \ref{cnrx1}.

\textbf{Direct method:} On substituting $ r_i=r,\ \forall i  $ and $ R_{ij}=R,\ \forall i,j,\  i\neq j$ in \eqref{eqnxr} for $ N=4 $ gives
\begin{align}
	\lhp{}(s,\x_1\mid \{\x_j\}_{j=2}^4)&=\frac{\lhpone(s,r)\left(s\lhpone(s,R)-1\right)}{3s^2\lhpone(s,R)^2-2s\lhpone(s,R)-1}\nonumber\\
	&=\frac{\lhpone(s,r)}{1+3s\lhpone(s,R)}
	=\lhpone(s,r)\times\sum_{n=0}^{\infty}(-3)^n\left(s\lhpone(s,R)\right)^n.
\end{align}
Substituting the value of $ \lhpone(s,x),\ x\in\{r,R\} $ in the above equation using \eqref{lapeq} and taking its inverse Laplace transform gives Corollary \ref{cnrx1}.
  \section{Hitting probability for a SIMO system with $N$ UCA of FARs}\label{ConjPr}
  
  For each FAR in UCA, $ r_i=r $, and $ \lhp{i}(s)=\lhp{j}(s) $. Substituting these values in \eqref{eqsnxr},  simplifying further and applying inverse Laplace transform, we get
  	\begin{align}
  		\hp{i}(t)&=\invlaplace{}{
  			\frac{\lhpone(s,r)}{1+\sum_{j=1,\ j\neq i}^{N}s\lhpone(s,R_{ji})} }{}\stackrel{(a)}{=}\invlaplace{}{
  			\frac{\lhpone(s,r)}{1+\sum_{m=1}^{\delta-1}2s\lhpone(s,R_m)+bs\lhpone(s,R_\delta)} }{},\label{lapucan}
  	\end{align}
  	where $ b=1 $ for even $ N $ and $ b=2$ for odd $ N $. $ (a) $ is due to $ R_i=R_{ji}=R_{j(N-i)}, \ 1\leq i\leq \delta $. Recall that $ R_m $ is the distance to $ m $th closest neighbor.\\
  	For odd N, \eqref{lapucan} can be further simplified as
  {\small	\begin{align}
  		\hp{1}(t)&=\invlaplace{}{
  			\frac{\lhpone(s,r)}{1+\sum_{i=1}^{\delta}2s\lhpone(s,R_i)} }{}\label{lNrxodd}\\
  		&\stackrel{(b)}{=}\invlaplace{}{\lhpone(s,r)\times\sum_{n=0}^{\infty}(-2)^n\sum_{k_1+k_2+\cdots+k_\delta=n} {n \choose k_1, k_2, \ldots, k_m}
  			\prod_{m=1}^\delta \left(s\lhpone(s,R_m)\right)^{k_m}}.
  	\end{align}}
  	Taking the inverse Laplace transform of the above equation gives \eqref{NUCAO}.\\
  	For even N, \eqref{lapucan} can be further simplified as
  	{\small\begin{align}	
  		\hp{1}(t)&=\invlaplace{}{\frac{\lhpone(s,r)}{1+\sum_{i=1}^{\delta-1}2s\lhpone(s,R_i)+s\lhpone(s,R_\delta)} }\label{lNrxeven}\\
  		&\stackrel{(c)}{=}\invlaplace{}{\lhpone(s,r)\times\sum_{n=0}^{\infty}(-2)^{n-k_\delta}\mkern-32mu\sum_{k_1+k_2+\cdots+k_\delta=n} {n \choose k_1, k_2, \ldots, k_m}
  			\prod_{m=1}^\delta \left(s\lhpone(s,R_m)\right)^{k_m}}.
  	\end{align}}
  	Taking the inverse Laplace transform of the above equation gives \eqref{NUCAE}. In the above, steps $ (b)  $ and $ (c) $ are obtained by applying multinomial theorem, \ie 
  	$ (x_1 + x_2  + \cdots + x_m)^n
  	= \sum_{k_1+k_2+\cdots+k_m=n} {n \choose k_1, k_2, \ldots, k_m}
  	\prod_{i=1}^m x_i^{k_i} $. 
  
\section{Signal gain of a two and three FAR system}\label{adg2rx}
  \textbf{$ 2-$FAR case}: 
  For a $ 2-$FAR system, the signal gain is given by
  \begin{align}
  	g(t) = \frac{2\hp{1}(t)}{\hpone(t,r)},\label{dg2rx}
  \end{align}
  When the FARs are arranged as UCA and $ r_i=r $, \eqref{lp2rx} becomes
  \begin{align}
  	\lhp{1}(s)   &=\lhpone(s,r)\frac{1}{1+s\lhpone(s,R)} \stackrel{(a)}{=}\lhpone(s,r)\times\sum_{n=0}^{\infty}(-1)^n\left(s\lhpone(s,R)\right)^n.\label{lap2rx}
  \end{align}
  Note that $ (a) $ is from the identity $ (1+x)^{-1}=\sum_{n=0}^{\infty}(-1)^nx^n $. Substituting \eqref{lapeq} in \eqref{lap2rx} and taking the inverse Laplace transform gives 
  \begin{align}
  	\hp{1}(t) &=\frac{a}{r}\sum_{n=0}^{\infty}\frac{(-a)^n}{R^n}\erfc\left(\frac{r-a+n(R-a)}{\sqrt{4Dt}}\right).\label{eq2rxhp1}
  \end{align}
  Rearranging \eqref{eq2rxhp1} gives
  \begin{align}
  	&\hp{1}(t)=\hpone(t,r)-\lambda(t,r,1).\label{eq2rxhp2}
  \end{align}
  Now substituting \eqref{eq2rxhp2} and \eqref{eq1rx} in \eqref{dg2rx} gives the desired result. 
  
  \textbf{$ 3-$FAR case}:
 The derivation  of $ g(t) $ for $ N=3 $ FAR case is similar to that of $ g(t) $ for $ N=2 $ FAR case given above. Rearranging the terms in \eqref{UCA3} and substituting it in \eqref{eqgain} gives the desired result.

	\bibliographystyle{IEEEtran}
	\bibliography{nrx_ref}
\end{document}